\documentclass[11pt]{article}
\usepackage[margin=1in]{geometry}
\usepackage{enumerate}
\usepackage{graphicx}
\usepackage{subfigure}
\usepackage[colorlinks=true,breaklinks=true,bookmarks=true,urlcolor=blue,
     citecolor=blue,linkcolor=blue,bookmarksopen=false,draft=false]{hyperref}
\usepackage{amsthm,amsmath,amssymb}
\def\EMAIL#1{\href{mailto:#1}{#1}}

\newcommand{\maxmab}{\mbox{\sc MaxMAB}}
\newcommand{\excess}{\mbox{\sc Excess}}

\newcommand{\lpmaxmab}{\mbox{\sc LPMaxMab}}
\newcommand{\modmaxmab}{\mbox{\sc LagMaxMab}}
\newcommand{\lpbud}{\mbox{\sc LPBud}}
\newcommand{\budlag}{\mbox{\sc BudLag}}

\newcommand{\achievedmaxmab}{\mbox{\sc Max-Pays-Mab}}

\newcommand{\onlythemax}{\mbox{\sc Only-Max-Costs}}
\newcommand{\allpay}{\mbox{\sc All-Plays-Cost}}
\newcommand{\oneatatime}{\mbox{\sc One-at-a-Time}}
\newcommand{\simult}{\mbox{\sc Simultaneous-Feedback}}

\newcommand{\lpone}{\mbox{\sc LP1}}
\newcommand{\lplag}{\mbox{\sc LPLag}} 
\newcommand{\lpdelay}{\mbox{\sc LPDelay}}
\newcommand{\h}[1]{\widehat{#1}}
\renewcommand{\t}[1]{\widehat{#1}}
\newcommand{\edge}[1]{\mbox{\sc \#Edges}{#1}}

\newtheorem{proposition}{Proposition}
\newtheorem{definition}{Definition}
\newtheorem{theorem}{Theorem}
\newtheorem{lemma}[theorem]{Lemma}
\newtheorem{claim}[theorem]{Claim} 
\newtheorem{corollary}[theorem]{Corollary}

\begin{document}

\title{Approximation Algorithms for Bayesian Multi-Armed Bandit Problems\footnote{This paper presents a unified version of results that first appeared in three conferences: STOC '07~\cite{GM07}, ICALP '09~\cite{GM09}, and APPROX '13~\cite{GM13}, and subsumes unpublished manuscripts~\cite{GM08,GMP10}.}}
\author{Sudipto Guha \thanks{Department of Computer and Information Sciences,
    University of Pennsylvania, Philadelphia, PA 19104.
\EMAIL{sudipto@cis.upenn.edu}. Research supported by NSF awards CCF-0644119, CCF-1117216. Part of this research was performed when the author was a visitor at Google.}
\and Kamesh Munagala
\thanks{Department of Computer Science,
    Duke University, Durham, NC 27708-0129.,
\EMAIL{kamesh@cs.duke.edu}.
Supported by an Alfred P. Sloan Research  Fellowship, an award from Cisco, and by NSF  grants CCF-0745761, CCF-1008065, and IIS-0964560.}
}

\newcommand{\R}{{\mathbb R}}
\renewcommand{\H}{{\cal H}}
\renewcommand{\O}{{\cal O}}
\newcommand{\comment}[1] {}
\newcommand{\etal}{{et al.~}}
\newcommand{\roundup}[1]{\left \lceil #1 \right \rceil}
\newcommand{\rounddown}[1]{\left \lfloor #1 \right \rfloor}
\newcommand{\expec}[1]{E \left ( #1 \right ) }
\newcommand{\rk}{\mbox{\sc Rk}}
\newcommand{\prob}[1]{\Pr \left ( #1 \right ) }
\newcommand{\hatyi}{\hat{y}_i}
\newcommand{\hatk}{\hat{k}} 
\newcommand{\delc}{{\Delta_C(i)}}
\newcommand{\alg}{{\cal A}}
\newcommand{\obj}{{\cal O}}
\newcommand{\gap}{\mbox{\sc Gap}\ }
\newcommand{\y}{{\mathbf Y}}
\newcommand{\C}{{\mathcal C}}
\newcommand{\X}{{\mathbf X}}
\newcommand{\bin}{{\mathbf{Bin}}}
\renewcommand{\L}{\mathcal{L}}
\newcommand{\eat}[1]{}
\newcommand{\acc}{E}
\newcommand{\uhull}{\mbox{\sc Uhull}\ }
\newcommand{\spike}{\mbox{\sc Spike}\ }
\newcommand{\stego}{\mbox{\sc Stego}\ }
\renewcommand{\k}{\kappa}
\newcommand{\A}{{\cal A}}
\newcommand{\I}{{\cal I}}
\newcommand{\D}{{\cal D}}
\newcommand{\E}{{\mathbf E}}
\renewcommand{\P}{{\cal P}}
\renewcommand{\S}{{\cal S}}
\newcommand {\T}{\mathcal T}
\newcommand{\lpopt}{\gamma^*}
\newcommand{\lpfeas}{{{\mathcal L}^{\mathcal O}}}

\renewcommand{\E}{\mathbf{E}}
\newcommand{\Pp}{{{\cal P}}}
\newcommand{\n}{{\overrightarrow{n}}}
\maketitle

\begin{abstract}
  In this paper, we consider several finite-horizon Bayesian
  multi-armed bandit problems with side constraints. These constraints
  include metric switching costs between arms, delayed feedback about
  observations, concave reward functions over plays, and
  explore-then-exploit models.  These problems do not have any known
  optimal (or near optimal) algorithms in sub-exponential running
  time; several of the variants are in fact computationally
  intractable (NP-Hard). All of these problems violate the exchange
  property that the reward from the play of an arm is not contingent
  upon when the arm is played.  This separation of scheduling and
  accounting of the reward is critical to almost all known analysis
  techniques, and yet it does not hold even in fairly basic and
  natural setups which we consider here. Standard index policies
  are suboptimal in these contexts, there has been little analysis of such
  policies in these settings.

  We present a general solution framework that yields constant factor
  approximation algorithms for all the above variants. Our framework
  proceeds by formulating a weakly coupled linear programming
  relaxation, whose solution yields a collection of compact policies
  whose execution is restricted to a single arm. These single-arm
  policies are made more structured to ensure polynomial time
  computability of the relaxation, and their execution is then
  carefully sequenced so that the resulting global policy is not only
  feasible, but also yields a constant approximation. We show that the
  relaxation can be solved using the same techniques as for computing
  index policies; in fact, the final policies we design are very close
  to being index policies themselves.

  Conceptually, we find policies that satisfy an approximate version
  of the exchange property, namely, that the reward from a play
  does not depend on time of play to within a constant factor.
  However such a property does not hold on a per-play basis and only
  holds in a global sense: We show that by restricting the state
  spaces of the arms, we can find single arm policies that can be
  combined into global (near) index policies that satisfy the
  approximate version of the exchange property analysis in expectation.
  The number of different bandit problems that can be addressed by
  this technique already demonstrate its wide applicability.

\end{abstract}
\thispagestyle{empty}
\newpage
\pagenumbering{arabic}
\section{Introduction.}
\label{sec:intro}
In this paper, we consider the problem of iterated allocation of
resources, when the effectiveness of a resource is uncertain a priori,
and we have to make a series of allocation decisions based on past
outcomes. Since the seminal
contributions of Wald~\cite{Wald47} and Robbins~\cite{Robbins}, a vast
literature, including both optimal and near optimal solutions, has
been developed, see references in~\cite{book,CL,GGW}. 

Of particular interest is the celebrated Multi-Armed Bandit (MAB)
problem, where an agent decides on allocating resources between $n$
competing actions (arms) with uncertain rewards and can only take one
action at a time (play the arm). The play of the arm provides the
agent both with some reward as well as some further information about
the arm which causes the state of the arm to be updated. The objective
of the agent is to play the arms (according to some specified
constraints) for a time horizon $T$ which maximizes sum of the
expected rewards obtained in the $T$ steps.  The goal of the algorithm
designer (and this paper) in this case is to provide the agent with a
a decision policy. Formally, each arm $i$ is equipped with a state
space $\S_i$ and a {\em decision policy} is a mapping from the current
states of all the arms to an {\em action}, which corresponds to
playing a subset of the arms. The goal is to use the description of
the state spaces $\{\S_i\}$ and the constraints as input, and output a
decision policy which maximizes the objective of the agent. The
running time of the algorithm that outputs such a decision policy, and
the complexity of specifying the policy itself, should ideally be
polynomial in $\sum_i|\S_i|$, which is the complexity of specifying
the input. In this regard, of particular interest are {\em Index
  Policies} where each arm is reduced to a single {\em priority} value
and the priority values of the arms (indices) are combined using some
scheduling algorithm.  Index policies are significantly easier to
implement and conceptually reduce the task of designing a decision
policy to designing {\em single arm} policies. However optimum index
policies exist only in limited settings, see \cite{book,GGW} for
further discussion.

\medskip In recent years MAB problems are increasingly being used in
situations where the possible alternative actions (arms) are machine
generated and therefore the parameter $n$ is significantly large in
comparison to the optimization horizon $T$. This is in contrast to
the historical development of the MAB literature which mostly considered
few alternatives, motivated by applications such as medical treatments
or hypothesis testing. While several of those applications considered
large $|\S_i|$, many of those results relied on concentration of
measure properties derived from the fact that $n/T$ could be assumed to 
be vanishingly small in those applications.

The recent applications of MAB problems arise in advertising, content
delivery, and route selection, where the arms correspond to
advertisers, (possibly machine generated) webpages, and (machine
generated) routes respectively, and the parameter $n/T$ does not
necessarily vanish. As a result, the computational complexity of
optimization becomes an issue with large $n$. Moreover, even simple
constraints such as budgets on rewards, concavity of rewards render
the computation intractable (NP Hard) --- and these recent
applications are rife with such constraints. This fact forces us to
rethink bandit formulations in settings which were assumed to be well
understood in the absence of computational complexity considerations.
One such general setting is the {\em Bayesian Stochastic Multi-Armed
  Bandit} formulation, which dates back to the results of
\cite{ABG49,Bradt}.

\newcommand{\p}{\mathbf{p}}
\paragraph{Finite Horizon Bayesian Multi-armed Bandit Problem.} We
revisit the classical finite-horizon\footnote{We study the
  finite-horizon version instead of the more ``standard'' discounted
  reward version for two reasons: It makes the presentation simpler
  and easier to follow, and furthermore, in many applications the
  discounted reward variant is used mainly as an approximation (in the
  limit) to the finite horizon variant. We note that taking the limit
  can create surprises with regard to polynomial time tractability.}
multi-armed bandit problem in the Bayesian setting. This problem forms
the basic scaffolding for all the variants we consider in this paper,
and is described in detail in Section~\ref{prelim}. There is a set of
$n$ independent arms. Arm $i$ provides rewards in an $i.i.d.$ fashion
based on an parametrized distribution $D_i(\theta_i)$, where the
parameter $\theta_i$ is unknown. A prior distribution $\D_i$ is
specified over possible $\theta_i$; the priors of different arms are
independent.  At each step, the decision policy can play a set of arms
(as allowed by the constraints of the specific problem) and observe
the outcome of {\bf only} the set of arms that have been played. Based
on the outcomes, the priors of the arms that were played in the
previous step are updated using Bayes' rule to the corresponding
posterior distribution. The objective (most often) is to maximize the
expected sum of the rewards over the $T$ time steps, where the
expectation is taken over the priors as well as the outcome of each
play.\footnote{We also consider the situation where the $T$ are used
  for pure exploration, and the goal is to optimize the expected
  reward for the $T+1^{st}$ step only.}
  
The process of updating the prior information for an arm can be
succinctly encoded in a natural state space $\S_i$ for arm $i$, where
each state encodes the posterior distribution conditioned on a
sequence of observations from plays made to that arm. At each state,
the available actions (such as plays) update the posterior
distribution by Bayes' rule, causing probabilistic transitions to
other states, and yielding a reward from the observation of the play.
This process therefore is a special case of the classic finite horizon
multi-armed bandit problem~\cite{GJ74}. The key property of the state
space corresponding to Bayesian updating is that the rewards satisfy
the {\em martingale property}: The reward at the current state is the
same as the expected reward of the distribution over next states when
a play is made in that state. The Bayesian MAB formulation is
therefore a canonical example of the {\em Martingale Reward Bandit}
considered in recent literature such as \cite{farias}\footnote{While
  most of the results in this paper translate to the larger class of
  Martingale Reward bandits, we focus the discussion on Bayesian
  Bandits in the interest of simplicity.}.

\smallskip
In this paper, we study several Bayesian MAB variants which are
motivated by side-constraints which arise both from modern and
historical applications. We briefly outline some of these constraints
below, and discuss technical challenges that arise later on.

\eat{ The state space of an arm $\S_i$ defines a natural DAG,
the root of which corresponds to the initial prior $\D_i$.  Every
other state $u \in \S_i$ corresponds to a set of observations, and
hence to the posterior $X_{iu}$ obtained by applying Bayes rule to
$\D_i$ with the those observations.\footnote{Since the underlying true
  object is a distribution, the specific sequence of observation is
  not important.} Playing the arm in state $u$ yields a transition to
state $v$ with probability $\p_{uv}$, provided $v$ can be obtained
from $u$ in one additional observation; the probability $\p_{uv}$ is
simply the probability of this observation conditioned on the
posterior $X_{iu}$ at $u$. The expected posterior mean at a state $u
\in \S_i$, denoted by $r_u = \E[X_{iu}]$ satisfies a Martingale
property $r_u = \sum_v \p_{uv}r_v$. We present an example of priors
and state spaces in Section~\ref{prelim}.}

\begin{enumerate}[(a)]
\item Arms can have {\em budgets} on how much reward can be accrued from them - this 
is a very natural constraint in most applications of the bandit setting.\label{budgetlabel}
\item The decision to stop playing an arm can be {\em irrevocable} due to high setup costs (or destruction) of the availability of the underlying action.\label{irrevocablelabel}
\item There could be a {\em switching cost} from one arm to another; these could be economically 
motivated and adversarial in nature. These could also arise from feasibility constraints 
on policies 
for instance, energy consumption in a sensor network to switch measurements from one node to the next, or for a robot to move physically to a different location.
\label{metriclabel}
\item Feedback about rewards can have a {\em time-delay} in arriving, for instance in pay-per-conversion auction mechanisms where information about a conversion arrives at a later time. These delays can be non-uniform.\label{delayedlabel}
\item The reward at a time-step can be a {\em non-convex} function of
  the set of arms played at that step.  Consider for instance, when a
  person is shown multiple different advertisements, and the sale is
  attributed to the ``most influential'' advertisement that was shown;
  or consider a situation where packets are sent across multiple links
  (arms) and delivering the packet twice does not give any additional
  reward. \label{maxmablabel}
\item An extreme example of (\ref{maxmablabel}) is the futuristic
  optimization where the the actions taken over the $T$ steps are for
  exploration only. The goal of the optimization to maximize the
  reward of the action taken at the $T+1^{st}$ step.
\label{futurelabel}
\end{enumerate}

\paragraph{Violation of Exchange Properties.} While the constraints
outlined above appear to be very different, they all outline a
fundamental issue: They all violate the property that the reward of
the arm also does not change depending of {\em when} it is
scheduled for play. This property of exchange of plays is the most
known application of the {\em idling bandit property}, which ensures
that the state of an arm does not change while that arm is not being
played. As a consequence an arm which has high reward on the current
play, can be played immediately (``exchanged'' with its later play)
without any loss of reward. This ability of exchanging plays, is at
the core of all existing analysis of stochastic MAB problems with provable
guarantees \cite{CL,book,GGW}. In fact index policies provide the
sharpest example of such an use of the exchangeable property.
Contingent on this property, the scheduling decisions can be decoupled
from the policy decisions --- as a result the
optimization problem can be posed without resorting to a ``time indexed''
formulation.

\smallskip However this exchangeability of plays does not hold in {\em
  all} the problems we consider here.  For example, in constraint
(\ref{maxmablabel}), the reward is non-convex combination of the
outcomes of simultaneous plays of different arms, the reward from the
play of an arm is not just the function of the current play of the
arm, because the other arms may have larger or smaller values.  This
is also obviously true if the goal is to maximize the reward of the
$T+1^{st}$ step, as in constraint (\ref{futurelabel}).  The same issue
arises in constraint (\ref{delayedlabel}) -- the information derived
about an arm is not a function of the current play of the arm; since
we may be receiving delayed information about that arm due to a
previous play. A similar phenomenon occurs in constraints
(\ref{metriclabel}) and (\ref{irrevocablelabel}) where conditioned on
the decision of switching to another arm, the (effective) reward from
an arm changes even though we have not played the arm.
For the above problems, the traditional and well understood index
policies such as the Gittins index~\cite{Gittins,Tsitsiklis} which are
optimal for discounted infinite horizon settings, are either undefined
or lead to poor performance bounds. Therefore natural questions that 
arise are: {\em Can we design
  provably near optimal policies for these problems? Can such policies be shown to be 
  almost (relaxations of) index policies? } And perhaps
more importantly: {\em What new conceptual analysis idea can we bring
  to bear on these problems?} Answering these questions is the
focus of this paper.

\subsection{A Solution Recipe} 
Our main contribution is to present a general solution recipe for {\em
  all} the above constraints.  At a high level, our approach uses
linear programming of a type that is similar to the widely used {\em
  weakly coupled} relaxation for bandit problems\footnote{We owe the
  terminology ``weakly coupled'' to Adelman and
  Mersereau~\cite{Adel}.} or ``decomposable'' linear program (LP)
relaxation~\cite{book,whittle}. This LP relaxation provides a
collection of {\em single-arm} policies whose execution is confined to
one bandit arm. Such policies, though not feasible, are desirable for
two reasons: First, they are efficient to compute since the
computation is restricted to one bandit arm (in contrast, note that
the joint state space over all the arms is exponential in $n$); and
secondly, they often naturally yield index
policies.  However, as mentioned before, prior results in this realm
crucially require indexability, and fail to provide any guarantees for
non-indexable bandit problems of the type we consider. Our recipe
below provides a novel avenue for efficiently designing and analyzing feasible
policies, albeit at the cost of optimality. 

\begin{description}
\item[Single-arm Policies:] The first step of the recipe is to
  consider the execution of a global policy restricted to a single
  arm, and express its actions only in terms of parameters related to
  that arm. This defines a {\em single-arm} policy. We identify a
  tractable state space for single arms, which is a function of the
  original state spaces $\S_i$, and reason that good single arm
  policies exist in that space. Though classic index policies are
  constructed using this approach~\cite{nino,whittle}, our approach is
  fundamentally different: In a global policy, actions pertaining to a
  single arm need not be contiguous in time -- our single-arm policies
  make them contiguous. We show that the {\em martingale property} of
  the rewards enables us to do this step, and this will be the crux of
  the subsequent algorithm. In essence, though these policies are {\em
    weaker} in structure than previously studied policies, this lack
  of structure is critical in deriving our results.

\item[Linear Program:] We write a linear programming relaxation, which
  finds optimal single-arm policies feasible for a small set of {\em
    coupling constraints} across the arms, where the coupling
  constraints describe how these single arm policies interact within a
  global policy. This step can either be a direct linear program;
  however, this is infeasible in many cases, and we show a novel
  application of the Lagrangian relaxation to not only solve this
  relaxation, but also analyze the approximation guarantee the
  relaxation produces.

\item[Scheduling:] The single-arm policies produced by the relaxation
  are not a feasible global policy. Our final step is to {\em
    schedule} the single arm policies into a globally feasible policy.
  This scheduling step needs ideas from analysis of approximation
  algorithms, and specially from the analysis of adaptivity gaps. We
  note that the feasibility of this step crucially requires the weak
  structure in the single-arm policies alluded to in the first item.
  The final policies we design are indeed very similar to (though not
  the same as) index policies; we highlight this aspect in
  Section~\ref{sec:gittins}.
\end{description}

We note that each of the above steps requires new ideas, some of which
are problem specific. For instance, for the first step, it is not
always obvious why the single-arm policies have polynomial complexity.
In constraint (\ref{maxmablabel}), the reward depends on the specific
set of $K$ arms that are played at any step.  This does not correspond
naturally to any policy whose execution is restricted to a {\em single
  arm}.  Similarly, in constraint (\ref{delayedlabel}), the state for
a {\em single arm} depends on the time steps in the past where plays
were made and the policy is awaiting feedback -- this is exponential
in the length of the delay.  Our first contribution is to show that in
each case, there are different state space for the policy which has
size polynomial in $\sum_i |\S_i|, T$, over which we can write a
relaxed decision problem.

Similarly, the scheduling part has surprising aspects: For constraint
(\ref{maxmablabel}), the LP relaxation does not even encode that we
receive the max reward every step, and instead only captures the {\em
  sum of rewards} over time steps. Yet, we can schedule the single-arm
policies so that the max reward at any step yields a good
approximation. For in constraint (\ref{delayedlabel}), the single-arm
policies wait different amounts of time before playing, and we show
how to interleave these plays to obtain good reward.

\medskip Finally, there are fundamental roadblocks, above and beyond
the specific issues mentioned for each problem, that applies to almost
every problem of this genre. There are very few techniques that
provide good algorithm-independent upper bounds for these problems.
One such approach has been linear programming. But there is a
fundamental hurdle regarding the use of linear programs in the context
of MAB problems.

A linear program often (and for many of the cases described herein)
accounts for the rewards globally after all observations have been
revealed.  In contrast, a policy typically has to adapt as the
observations are revealed. The gap between these two accounting
processes is termed as the adaptivity gap (an example of such in the
context of scheduling can be found in \cite{DGV}).  The adaptivity gap
is a central facet of bandit problems where the policies are adapt to
the observable outcomes.  We show that there exists policies for whom
the adaptivity gap is at most a constant factor! Observe that showing
this property also implies satisfying an {\em approximate version of
  the exchange property, namely, up to a constant factor the reward
  from a play does not depend on when the play is made.} Note that
this property does not (can not) hold on a per play basis and is a
global statement which holds in expectation.

\subsection{Problem Definitions, Results, and Roadmap} 
We now describe the specific problems we study in this paper,
highlighting the algorithmic difficulty that arises therein.  In terms
of performance bounds, we show that the resulting policies have
expected reward which is {\em always} within a fixed constant factor
of the reward of the optimal policy, regardless of the parameters
(such as $n,T$) of the problem, and regardless of the nature of the
priors $\D_i$ of the arms. Such a multiplicative guarantee is termed
as {\em approximation ratio} or {\em factor}\footnote{Following
  standard convention, we will express this ratio as a number larger
  than $1$, meaning that it will be an upper bound on how much larger
  the optimal reward can be compared to the reward generated by our
  policy.}. We present constant factor approximation algorithms for
all the problems we consider. The running times we achieve are small
polynomial functions of $T$ and $\sum_i |\S_i|$, which is also the
complexity of an explicit representation of the state space of each
arm. As mentioned above, our running times are very close to the time
taken to compute standard indices.

Approximation algorithms are a natural avenue for designing policies,
since fairly natural versions of the problems we consider can be shown
to be {\sc NP-Hard}. For instance, for constraint (\ref{maxmablabel}),
where the per-step reward is the maximum observed value from the set
of plays is NP-Hard because it generalizes computing a subset of $k$
distributions that maximizes the expected maximum~\cite{GoelGM}. The
bandit problem with switching costs (constraint~\ref{metriclabel}) is
{\sc Max-SNP} Hard\footnote{There exists $\epsilon>0$ such that
  computing a $1+\epsilon$ approximation is NP Hard.} since it
generalizes an underlying combinatorial optimization problem termed
{\em orienteering}~\cite{BlumCKLMM03}.  For other problems, such as
the constraints of irrevocable policies
(constraint~\ref{irrevocablelabel}), there are examples where no
global policy satisfying the said constraint can achieve a reward
within a small constant factor of the upper bound proved by some
natural linear programming relaxation.  In this case, and many others
we provide alternate analysis of existing heuristics which have been
shown to perform well in practice~\cite{farias}.

We now discuss the specific problems. Each of the problems involve
budgets; however, these can be incorporated in the state space in a
natural way (by not allowing states corresponding to larger rewards).
In some cases, specially constraint (\ref{maxmablabel}), budgets pose
more unusual challenges.

\paragraph{Irrevocable Policies for Finite Horizon MAB Problem.} This
problem was formalized by Farias and Madan \cite{farias}.  This is the
Bayesian MAB problem described earlier where the policy plays at most
$K$ arms in a step and the objective is to maximize the expected
reward over a horizon of $T$ steps. Each of the arms have budgets on
the total reward that can be accrued from that arm.  There is an
additional constraint that we do not revisit any arm that we have
stopped playing. Using the algorithmic framework introduced
in~\cite{GM07}, Farias and Madan showed an $8$ approximation algorithm
which also works well in practice.  In Section~\ref{finite}, we
provide a better analysis argument (based on revisiting \cite{GM07} and newer ideas) 
and that argument improves the result to a
factor $(2+\epsilon)$-approximation for any $\epsilon>0$ in time
$O((\sum_i \edge(\S_i))\log (nT/\epsilon))$ where $\edge(\S_i)$ is the
number of edges that define the statespace $\S_i$.  The parameter
$\edge(\S_i)$ is a natural representation of the sparsity of the input
and thus we can view the algorithm to be almost linear time in that
said measure.  We also show that $2-O(\frac1n)$ is the best possible
bound against weakly coupled relaxations over single arm policies for
$n$ arms.  Although this is not the main result in this paper, we
present this problem first because the intuitive analysis idea is used throughout the paper and 
this problem has a right level of complexity for illustrative purposes.

\paragraph{Traversal Dependent MAB Problems.} It is widely
acknowledged~\cite{BS94,ML} that the scenarios that call for the
application of bandit problems typically have constraints/costs for
switching arms. Banks and Sundaram \cite{BS94} provide an illuminating
discussion even for the discounted reward version, and highlight the
technical difficulty in designing reasonable policies.  Since the
switching costs couple the arms together, it is not even clear how to
define a weakly coupled system. These constraints/costs can be
motivated by strategic considerations and even be adversarial:
Price-setting under demand uncertainty~\cite{Roth}, decision making in
labor markets~\cite{labor1,labor2,BS94}, and resource allocation among
competing projects~\cite{banks1,GGW}. For instance~\cite{BS94,banks1},
consider a worker who has a choice of working in $k$ firms. A priori,
she has no idea of her productivity in each firm. In each time period,
her wage at the current firm is an indicator of her productivity
there, and partially resolves the uncertainty in the latter quantity.
Her expected wage in the period, in turn, depends on the underlying
(unknown) productivity value. At the end of each time step, she can
change firms in order to estimate her productivity at different
places. Her payoff every epoch is her wage at the firm she is employed
in at that epoch. The added twist is that changing firms does not come
for free, and incurs a cost to the worker. What should the strategy of
the worker be to maximize her (expected) payoff?  A similar problem
can be formulated from the perspective of a firm trying out different
workers with a priori unknown productivity to fill a post. On the
other hand, the switching costs may be constraints imposed on the
policies by physical considerations -- for example an energy
constrained robot moving to specific locations. The above discussion
suggests two sets of problems -- the {\em
  Adversarial Order Irrevocable Policies for Finite Horizon MAB}
problem and the {\em MAB with Metric Switching Costs} problem. We
address each of them below:

\paragraph{Adversarial Order Irrevocable Policies for Finite Horizon
  MAB.}
In all these cases, the decision 
agent may not have the flexibility to choose an ordering of events and 
yet have to provide some algorithm with provable rewards.
 In this problem we are asked to design a collection of single
arm policies subject to a total horizon of $T$ and $K$ plays at a
time. After we design the policies, an adversary chooses the order in
which we can play the arms irrevocably -- if we quit playing an arm
then we cannot return to playing that same arm. We then schedule the
policies such that we maximize the expected reward. This problem
naturally models a ``switching cost'' behavior where the costs are
induced by an adversary. To the best of our knowledge this problem had
not been formulated earlier -- in Section~\ref{sec:metric} 
we provide a $(4+\epsilon)$-approximate
solution (when compared to the best possible order we could have
received) that can be computed in time $O((\sum_i \edge(\S_i))\log
(nT/\epsilon))$ as before. The main benefit of this problem is
revealed in the next problem on metric switching costs.

\paragraph{MAB with Metric Switching Costs.} In this variant, the arms are
located in a metric space with distance function $\ell$. If the
previous play was for arm $i$ and the next play is for arm $j$, the
{\em distance cost} incurred is $\ell_{ij}$. The total distance cost
is at most $L$ in all decision paths. The policy pays this cost
whenever it switches the arm being played. If the budget $L$ is
exhausted, the policy must stop. This problem already models the more
common ``switching in-and-out'' costs using a star graph.  To the best
of our knowledge there was no prior theoretical results on the
Bayesian multi-armed bandit with metric switching
costs\footnote{Except the conference paper \cite{GM09}, which is
  subsumed by this article.}. Observe that in absence of the metric
assumption, the problem already encodes significantly difficult graph
traversal problems and are unlikely to have bounded performance
guarantees.  In Section~\ref{sec:metric}, we present a $(4 +
\epsilon)$-approximation\footnote{These result weakens if we use the
  $4$-approximation in \cite{BlumCKLMM03} instead of the $2+\epsilon$
  approximation for the orienteering problem in \cite{CKP}. The
  constants become $16/3 + \epsilon$ and $25/4 +\epsilon$ for $K=1$
  and $K\geq 2$ respectively.}  (for any $\epsilon>0$) the
Finite-horizon MAB problem with metric switching costs for $K=1$. The
result worsens to $(4.5+\epsilon)$-approximation for any $K\geq 2$.
This result uses the ideas from the adversarially ordered bandits and separates 
the hard combinatorial traversal decisions from
the bandit decisions.

\paragraph{Delayed Feedback.} In this variant, the feedback about a
play (the observed reward) for arm $i$ is obtained after some
$\delta_i$ time steps.  This notion was introduced by
Anderson~\cite{Anderson} in an early work in mid 1960s.  Since then,
though there have been additional results~\cite{Suzuki,Choi}, a
theoretical guarantee on adaptive decision making under delayed
observations has been elusive, and the computational difficulty in
obtaining such has been commented upon
in~\cite{sched,Armitage,Simon,Eick}. Recently, this issue of delays
has been thrust to the fore due to the increasing application of
iterative allocation problems in online advertising. As an example
described in Agarwal etal in \cite{agarwal}, consider the problem of
presenting different websites/snippets to an user and induce the user
to visit these web pages. Each pair of user and webpage define an arm.
The delay can also arise from batched updates and
systems issues, for instance in information gathering and
polling~\cite{agarwal}, adaptive query processing~\cite{TianD03},
experiment driven management and profiling~\cite{expt,whatif},
reflective control in cloud services~\cite{cloud}, or unmanned aerial
vehicles~\cite{leny}. There has been little formal analysis of delayed
feedback and its impact on performance despite the fact that delays
are ubiquitous.

We note that in the case of delayed feedback, even for a single arm,
the decision policy needs to keep track of plays with outstanding
feedback for that arm, so that the optimal single-arm decision policy
has complexity which is exponential in the delay parameter $\delta_i$.
It is therefore not clear how to solve the natural weakly coupled LP
efficiently.  In a broader sense, the challenge in the delayed setting
is not just to balance explorations along with the exploitation, but
also to decide a schedule for the possible actions both for a single
arm (due to delays) as well as across the arms.  In
Section~\ref{sec:delay}, we show several constant factor
approximations for this problem. We show that when $\max_i \delta_i
\leq \sqrt{T}/50$ then the approximation ratio does not exceed $3$
even for playing $K$ different arms (with additive rewards) simultaneously. The approximation
ratio worsens but remains $O(1)$ for $\max_i \delta_{i} \leq T/(48 \log
T)$. Interestingly the policies obtained from these algorithm satisfy the
property for each single arm which is best summarized as: ``quit,
exploit or double the exploration'' which is a natural policy in
retrospect. Moreover the different algorithmic approaches expose the
natural ``regimes'' of the delay parameter from small to large.

\paragraph{Non-Convex Reward Functions.} Typically, MAB problems have
focused on either the reward from the played arm, or if multiple arms
are played, the sum of the rewards of these arms. However, in many
applications, the reward from playing multiple arms is some concave
function of the individual rewards. Consider for instance charging
advertisers based on tracking clicks made by a user before they make a
conversion; in this context, the charge is made to the most
influential advertiser on this path and is hence a non-linear function
of the clicks made by the user. As another example, consider an fault
tolerant network application wishing to choose send packets along $K$
independent routes~\cite{Akella}. The network only has prior
historical information about the success rate of each of the routes
and wishes to simultaneously explore the present status of the routes
as well as maximize the probability of delivery of the packet. This
problem can be modeled as a multi-armed bandit problem where multiple
arms are being played at a time step and the reward is a nonlinear
(but concave) function of the rewards of the played arms. 
As a concrete example, we define the \maxmab\ problem, where at most
$K$ arms can be played every step, and the reward at any step is the
{\em maximum} of the observed values.  This problem produces two
interesting issues: 

First, {\em if only the maximum value is obtained
  as a reward, should the budgets of the other arms be decreased?} A
case can be made for either an answer of ``yes'' (budgets indicate
opportunity) which defines a \allpay\ accounting or an answer of
``no'' (budgets indicate payouts of competing advertisers in a repeated
bidding scenario) which defines \onlythemax\ accounting.

Second, {\em if multiple
  arms are played and their rewards are non-additive, does the
  modeling of how the feedback is received (\oneatatime\ or
  \simult) affect us?}  In Section~\ref{sec:maxmab}, we
provide a $(4+\epsilon)$-approximation in the model where the budgets
of all arms are decreased and the feedback is received one-at-a-time.
We also show that these four variants (from the two issues) are
related to each other and their respective optimum solutions do not
differ by more than a $O(1)$ factor. We show this by providing
algorithms in a weaker model that achieve $O(1)$ factor reward of the
optimum in a correspondingly stronger model.  From a policy
perspective, the policies designed for this problem have the property
that restricted to a single arm, the differences between small reward
outcomes are ignored. Therefore the policy focuses on an almost
index-like computation of only the large reward outcomes.

\paragraph{Future Utilization and Budgeted Learning.} One of the more
recent applications of MAB has been the Budgeted Learning type
applications popularized by \cite{naive,madani,moore}. This
application has also been hinted at in \cite{ABG49}. The goal of a
policy here is to perform pure exploration for $T$ steps. At the
$T+1^{st}$ step, the policy switches to pure exploitation and
therefore the objective is to optimize the reward for the $T+1^{st}$
step. The simplest example of such an objective is to choose one arm
so as to maximize the {\em (conditional) expected value}; more
formally, given any policy $\pi$ we have a probability distribution
over final (joint) state space. In each such (joint) state, the policy chooses
the arm that maximizes the expected value -- observe that this is a
function $g(\pi)$ of the chosen policy.  The goal is to find the
policy that maximizes the expected value of $g(\pi)$ where the
expectation is taken over all initial conditions and the outcomes of
each play.  Note that the non-linearity arises from the fact that
different evolution paths correspond to different final choices.
Natural extensions of such correspond to Knapsack type constraints and
similar problems have been discussed in the context of power
allocation (when the arms are channels) \cite{cover} and optimizing
``TCP friendly'' network utility functions \cite{LL99}. In
Section~\ref{sec:budget} we provide a $(3+\epsilon)$-approximation 
for the basic variant\footnote{Subsuming the main result of the conference paper
  \cite{GM07}.}. 

\subsection{Related Work}
\label{related}
Multi-armed bandit problems have been extensively studied since their
introduction by Robbins in \cite{Robbins}. From that starting point
the literature has diverged into a number of (often incomparable)
directions, based on the objective and the information available. In
context of theoretical results, one typical
goal~\cite{LaiRobbins,Auer,CL,Golovin,maxkmab,experts,KalaiV03,FlaxmanKM05}
has been to assume that the agent has absolutely no knowledge of
$\rho_i$ (model free assumption) and then the task has been to
minimize the ``regret'' or the lost reward, that is comparing the
performance to an omniscient policy that plays the best arm from the
start. However, these results both require the reward rate to be large
and large time horizon $T$ compared the number of arms (that is, vanishing $n/T$). In the
application scenarios mentioned above, it will typically be the case
that the number of arms is very large and comparable to the optimization horizon and the 
reward rates are low.

Moreover almost all analysis of stochastic MAB in the literature has
relied on the {\bf exchange property}; these results depend on
``moving plays/exchange arguments'' --- if we can play arm $i$
currently, we consider waiting or moving the play of the arm to a
later time step without any loss.  The exchange properties are required
for defining submodularity as well as its extensions, such as sequence
or adaptive submodularity \cite{Golovin,golovin2,saeed}.
The exchange property is {\bf not} true in cases
of irrevocability (since we cannot have a gap in the plays, so moving
even a single play later creates a gap and is infeasible), delays
(since if the next play of an arm is contingent on the outcome of the
previous play, we cannot move the first play later arbitrarily because
then we may make the second play before the outcome of the the first
is known, which is now infeasible), or non-linear rewards (since the
reward is explicitly dependent on the subset of plays made together).
It may be appear that we can avoid the issue by appealing to
non-stochastic/adversarial MABs~\cite{nonstochastic} -- but they do
not help in the presence of budget constraints which couples the
action across various time steps.  It is known that online convex
analysis cannot be analyzed well in the presence of large number of
arms and a ``state'' that couples the time steps~\cite{mkearns}.
Budgets are one of the ways a ``state'' is enforced. Delays and
irrevocability also explicitly encode state. All future utilization
objectives trivially encode state as well.

The problem considered in \cite{maxkmab} is similar in name but is
different from the \maxmab\ problem we have posed herein --- that
paper maximizes the single maximum value seen across all the arms and
across all the $T$ steps.   The authors of \cite{GKN} show that several
natural index policies for the budgeted learning problem are constant
approximations using analysis arguments which are different from those
presented here. The specific approximation ratios proved in~\cite{GKN}
are improved upon by the results herein with better running times. The authors of \cite{Gupta}
present a constant approximation for non-martingale finite horizon
bandits; however, these problems require techniques that are
orthogonal to those in this paper. The problems considered in~\cite{GMS10} is
an infinite horizon restless bandit problem. Though that work also
uses a weakly coupled relaxation and its Lagrangian (as is standard
for many MAB problems), the techniques we use here are different.

\section{Preliminaries: State Spaces, Budgets and Decision Policies}
\label{prelim}
Recall the basic setting of the finite horizon Bayesian MAB problem.
There is a set of $n$ independent arms. Arm $i$ provides rewards in an
$i.i.d.$ fashion from a parametrized distribution $D_i(\theta_i)$. The
parameter $\theta_i$ can be arbitrary. It is unknown at the outset,
but a prior $\D_i$ is specified over possible values of $\theta_i$.
The priors of different arms are independent.  At each step, a
decision policy can play a set of arms (as allowed by the constraints
of the specific problem) and observe the outcome of {\em only} the set
of arms that have been played. Based on the outcomes, the priors
$\{\D_i\}$ of the arms that were played in the previous step are
updated using Bayes' rule to the corresponding posterior
distributions.  The objective of the decision policy (most often) is
to maximize the expected sum of the rewards obtained over the $T$ time
steps, where the expectation is taken over the priors as well as the
outcome of each play. We note that other objectives are possible; we
also consider the case where the $T$ plays are used for pure
exploration, and the goal is to optimize the expected reward for the
$T+1^{st}$ play.

\subsection{State Space of an Arm and Martingale Rewards}
Let $\sigma$ denote a sequence of observed rewards, when $k =
|\sigma|$ plays of arm $i$ are made. Given these observations, the
prior $\D_i$ uniquely updates to a posterior $\D_{i\sigma}$. Since the
rewards are $i.i.d.$, only the {\em set} (as opposed to sequence) of
observations $\sigma$ matters in uniquely defining the posterior
distribution. Suppose $D_i(\theta)$ is a $r$ valued distribution, a
sequence of $k$ plays can therefore yield ${r + k \choose k}$ sets
$\sigma$. Each of these uniquely defines a posterior distribution
$\D_{i\sigma}$.  We term each of these sets $\sigma$ as a {\em state}
of arm $i$. At any point in time, the arm $i$ is in some state, and
over a horizon of $T$ steps, the number of possible states is
$O(T^r)$.  We denote the set of these states as $\S_i$, and a generic
state as $u \in \S_i$. Let the start state corresponding to the
initial prior $\D_i$ be denoted $\rho_i$.

Given a posterior distribution $\D$ of arm $i$ (corresponding to some
state $u \in \S_i$ and set of observations $\sigma$), the next play
can be viewed as follows: A parameter $\theta_i$ is drawn from $\D$,
and a reward is observed from distribution $D_i(\theta_i)$. Therefore,
the posterior distribution $\D$ at state $u$ also defines the
distribution over observed rewards of the next play. If the observed
reward is $q$, this modifies the state from $u$ with set of
observations $\sigma$ to state $v$ corresponding to the set $\sigma
\cup \{q\}$. We define $\p_{uv}$ as the probability of this event
conditioned on being at state $u$ and executing a play.

We further define $r_u = \E[D_i(\theta) | \S_i = u]$ as the expected
value of the observed reward of this play. We term this the {\em
  reward} of state $u$.
Let $\edge(\S_i)$ be the number of {\em non-zero transition edges} in
$\S_i$, i.e., $\p_{uv} \neq 0$. We use the notation $|\S_i|$ to denote the number of states
and if the play in each state has at most $d$ outcomes then 
$\edge(\S_i) \leq d |\S_i| \leq  |\S_i|^2$. Observe that $\sum_i
\edge(\S_i)$ is a natural measure of input sparsity.

\paragraph{Compact Representation and Martingale Property.}
At this point, we can forget about Bayesian updating (for the most part), and pretend we have a state space $\S_i$ for arm $i$, where state $u$ yields reward $r_u$, and has transition probability matrix $\{\p_{uv}\}$. The goal is to design a decision policy for playing the arms (subject to problem constraints) so that the expected reward over $T$ steps is maximized. Since the prior updates satisfy Bayes' rule, it is a standard observation that $r_u = \sum_{v \in \S_i} \p_{uv}r_v$ for all states $u \in \S_i$. In other words, the state space of an arm satisfies the {\em Martingale property} on the rewards. We note that all our results except for the one in Section~\ref{sec:delay} hold for arbitrary state spaces satisfying the martingale property; however, we choose to present the results for case of Bayesian updates, since this is the canonical and most widely known application of the MAB problems.

\paragraph{Example: Bernoulli Bandits.} 
Consider a fixed but unknown distribution over two outcomes (success
or failure of treatment, click or no click, conversion or no
conversion). This is a Bernoulli$(1,\theta)$ trial where $\theta$ is
unknown.  A family of prior distribution on $\theta$ is the Beta
distribution; this is the conjugate prior of the Bernoulli
distribution, meaning that the posterior distributions also lie within
the same family.  A Beta distribution with parameters
$\alpha_0,\alpha_1 \in \mathbb{Z}^+$, which we denote
$Beta(\alpha_1,\alpha_0)$ has p.d.f. of the form $c\theta^{\alpha_1-1}
(1-\theta)^{\alpha_0-1}$, where $c$ is a normalizing constant.
$Beta(1,1)$ is the uniform distribution.  The distribution
$Beta(\alpha_1,\alpha_0)$ corresponds to the current (posterior)
distribution over the possible values of $\theta$ after having
observed $(\alpha_1-1)$ $1$'s and $(\alpha_0-1)$ $0$'s, starting from
the belief that $\theta$ was uniform, distributed as $Beta(1,1)$.

Given the distribution $u = Beta(\alpha_1,\alpha_0)$ as the prior, the
expected value of $\theta$ is $\mu_u =
\frac{\alpha_1}{\alpha_1+\alpha_0}$, and this is also the expected
value of the observed outcome of a play conditioned on this prior.
Updating the prior on a sample is straightforward. On seeing a $1$,
the posterior (of the current sample, and the prior for the next
sample) is $v = Beta(\alpha_1+1,\alpha_0)$, and this event happens
with probability $\mu_u$. On seeing a $0$, the new distribution is $w
= Beta(\alpha_1,\alpha_0+1)$ and this event happens with probability
$1-\mu_u$.

The posterior density $u = Beta(\alpha_1,\alpha_0)$ corresponds to
state $u \in \S_i$, and can be uniquely specified by the values
$(\alpha_1,\alpha_0)$. Therefore, $r_u = \mu_u$. If the arm is played
in this state, the state evolves to $v$ with probability $\p_{uv} =
\mu_u$ and to $w$ with probability $\p_{uw} = 1-\mu_u$.
As stated earlier, the hard case for our algorithms (and the case that
is typical in practice) is when the input to the problem is a set of
arms $\{i\}$ with priors $\D_i \sim Beta(\alpha_{1i},\alpha_{0i})$
where $ \alpha_{0i} \gg \alpha_{1i}$ which corresponds to a set of
poor prior expectations of the arms.
Observe that $\edge(\S_i) \leq 2 |\S_i|$ for this example.

\subsection{Modeling Budgets and Feedback}
\label{budgets}
One ingredient in our problem formulations is natural budget constraints in individual arms. As we will see later, these constraints introduce non-trivial complications in designing the decision policies. There are three types of budget constraints we can consider.
\begin{description}
\item[Play Budget Model:] This corresponds to the number of times an individual arm $i$ can be played, typically denoted by $T_i$ where $T_i \le T$. In networking, each play attempts to  add traffic to a route and we may choose to avoid congestion. In online advertising, we could limit the  number of impressions of a particular advertisement, often  done to limit advertisement-blindness.
\item[\allpay\ Model] The total reward obtained from arm $i$ should be at most its budget $B_i$. The arm can be played further, but these plays do not yield reward.
\item[\onlythemax\ Model]  In {\sc MaxMab}, the reward of only the arm that is maximum is used, and hence only this budget is depleted. In other words, there is a bound on the total  reward achievable from an arm (call this $A_i$), but we only count reward from an arm when it was the maximum value arm at any step.  
\end{description}

\noindent Observe that the play budget and \allpay\ models simply involve
truncating the state space $\S_i$, so that an arm cannot be played if
the number of plays or observations violates the budget constraint --
this obviates the need to discuss this constraint explicitly in most
sections. The \onlythemax\ model is more complicated to handle, requiring 
expanding the state space to
incorporate the depleted budget. We discuss that model only 
in Section~\ref{hardest}. Note however that the \allpay\
model is also nontrivial to encode in the context of delays in
feedback, and we discuss that model in Section~\ref{sec:delay}
further.

\paragraph{}Modeling the feedback received is an equally important
issue. We just alluded to delays in the feedback for the play of a
single arm. However if multiple plays are being made simultaneously -- 
the definition of simultaneity has interesting consequences. We can conceive two models:
\begin{description}
\item[\oneatatime\ {\sc Feedback}:] In this model even if we are allowed to play $K>1$ arms in a 
single time slot, we make only one play a time and receive the feedback before making the play of a different arm in that same time slot.
\item[\simult:] In this model we decide on a palette of at most $K$ plays in a time slot 
and make them before receiving the result of any of the plays in that time slot.
\end{description}
Clearly, the \oneatatime\ model allows for more adaptive policies.
However for additive rewards, where we sum up the rewards obtained
from all the plays made in a time slot -- this issue regarding the
order of feedback is not very relevant. However in the case where we
have subadditive rewards then the distinction between \oneatatime\ and
\simult\ is clearly significant -- specially in the \allpay\
accounting. The distinction continues to hold in the \onlythemax\ pays
model as well, because we can adapt the set of arms to play to the
outcome of the initial plays in the same time slot.

\subsection{Decision Policies and Running Times}
A decision {\bf policy} $\P$ is a mapping from the current state of
the overall system to an {\em action}, which involves playing $K$ arms
in the basic version of the problem. This action set can be richer, as
we discuss in subsequent sections. Observe that the state of the
overall system involves all the knowledge available to the policy,
including:
\begin{enumerate}
\item The joint states $\{u \in \S_i\}$ of all the arms;
\item The current time step, which also yields the remaining time horizon;
\item The remaining budget of all the arms;
\item Plays made in the past for which feedback is outstanding (in the context of delayed feedback).
\end{enumerate}

\noindent 
Therefore, the ``state'' used by a decision policy could be much more complicated than the product space of $ \S_i$, and one of our contributions in this paper is to simplify that state space and make it tractable.

\paragraph{Running Times.} Our goal will be to compute decision
policies in time that depends polynomially on $T$ and $\sum_i |\S_i|$,
which is the sum of the sizes of the state spaces of all the arms. In
the case of Bernoulli bandits discussed above, this would be
polynomial in the number of arms $n$, and the time horizon $T$.  Our
running times and the description of the final policies will be
comparable (in a favorable way) to the time required for dynamic
programming to compute the standard index
policies~\cite{book,Gittins}. This will become clear in
Section~\ref{finite}; however, fine-tuning the running time is not the
main focus of this paper. However our running times will often be
linear in the natural input sparsity, that is, $\sum_i \edge(\S_i)$.

\subsubsection{Single Arm Policies}
A single arm policy $P_i$ for arm $i$ is a policy whose execution
(actions and rewards) are only confined to the state space of arm $i$.
In other words, it has one of several possible actions at any step; in
the basic version of the problem, the available actions at any step
involve either (i) play the arm -- if the arm is in state $u \in
\S_i$, this yields reward $r_u$; or (ii) stop playing and exit.
Furthermore, the actions of $P_i$ only depend on the state variables
restricted to arm $i$ -- the current state $u \in \S_i$, remaining
budget of arm $i$, the remaining time horizon, etc. In most 
situations we will eliminate states that cannot be reached in $T$ steps
(the horizon). Formally;

\begin{definition}
\label{singlearmdef}
Let $\S_i(T)$ be the state space $\S_i$ truncated to $T$ steps.
Let $\C_i(T)$ describe all single-arm policies of arm $i$ with horizon
$T$.  Each $P_i \in \C_i(T)$ is a (randomized) mapping from the
current state of the arm to an action. The {\em state} of the system
is captured by the current posterior $u \in \S_i(T)$ and the remaining
time horizon.  Such a policy at any point in time has at two available
actions: (i) Play the arm; or (ii) stop. 
Note that the states and actions are for the
basic problem, so that variants will have more states/actions which will be 
described at the appropriate context.
\end{definition}

\begin{definition}
\label{notationdef}
  Given a single-arm policy $P_i$ let $R(P_i)$ to be the expected
  reward and $\T(P_i)$ as the expected number of plays made. The
  expectation is taken over all randomizations within the policy as
  well as all outcomes of all plays.
\end{definition}

The construction of single-arm policies will be non-trivial, as will
be described in subsequent sections. One of the contributions of this
paper is in the development of single-arm policies with succinct
descriptions. Succinct descriptions also allow us to focus on key 
aspects of the policy and analyze simple and common modifications 
of the policies.

\section{Irrevocable Policies for Finite Horizon Bayesian MAB Problems}
\label{finite}

In this section we consider the basic finite horizon Bayesian MAB
problem (as defined in Section~\ref{sec:intro}) with the additional
constraint of {\em Irrevocability}; that is, the plays for any arm are
contiguous.  Recently, Farias and Madan~\cite{farias} provided an
$8$-approximation for this problem.  Previously, constant factor
approximations were provided in \cite{GKN,GM09} for the finite horizon
MAB problem (not necessarily using irrevocable policies). The analysis herein
provides a significant strengthening of all previous results on this
problem.  The contribution of this section is the introduction of a
palette of basic techniques which are illustrative of all the
(significantly more difficult) problems in this paper.

\medskip
\noindent 
In term of specific results we show that given a solution of the
standard LP formulation for this problem, a simple scheduling
algorithm provides a $2$ approximation (Theorem~\ref{unitmab}).
Moreover we prove that this is the best possible result against the
standard LP solution. We then show that we can compute a near optimal
solution of the standard LP efficiently, that is, for any $\epsilon \in (0,1]$ we provide an algorithm that runs in $O((\sum_i \edge(\S_i) )\log (nT/\epsilon))$ time
and provides a $2+\epsilon$-approximate
solution against the optimum solution of that standard LP
(Theorem~\ref{blahtheorem} and Corollary~\ref{blah2}).

In terms of techniques we introduce a compact LP representation and its
analysis based on Lagrangians, which are used throughout the rest of
the paper. Likewise we consider the technique of {\em truncation} 
in Section~\ref{part1} -- while this technique is similar to 
\cite[Lemma 2]{farias} but our presentation uses a slightly 
different intuition that is also useful throughout this paper.

\paragraph{Roadmap:} We present the standard LP relaxation that bounds
the value of the optimum solution $OPT$ in Section~\ref{weak1}; along
with an interpretation of fractional solution of that LP in terms of
single arm policies. We show how the standard LP can be represented in
a compact form demonstrating the {\em weak coupling}.  As stated, we
discussion truncation in Section~\ref{part1}. In Section~\ref{simple}
we then present the $2$-approximation result (Theorem~\ref{unitmab})
and the (family of) examples which demonstrate that the bound of $2$
is tight. In Section~\ref{sec:compact} we then show how to achieve the
$(1+\epsilon)$-approximation (Theorem~\ref{blahtheorem} and
Corollary~\ref{blah2}) for the compact LP. We conclude in
Section~\ref{sec:gittins} with a comparison between the Lagrangian
induced solutions and the Gittins Index.
 
\subsection{Linear Programming Relaxations and Single arm Policies} 
\label{weak1}
Consider the following LP, which has two variables $w_u$ and $z_u$ for
each arm $i$ and each $u \in \S_i(T)$. Recall $\S_i(T)$ is the
statespace $\S_i$ truncated to a horizon of $T$ steps.  The variable
$w_u$ corresponds to the probability (over all randomized decisions
made by any algorithm as well as the random outcomes of all plays made
by the same algorithm) that during the execution of arm $i$ enters
state $u \in \S_i$. Variable $z_u$ corresponds to the probability of
playing arm $i$ in state $u$.

{\small
\[ \begin{array}{rcll}
\displaystyle \lpone & = & \displaystyle \mbox{Max} \sum_{i=1}^n \sum_{u \in \S_i(T)} r_u z_u & \\ 
\sum_{i=1}^n   \sum_{u \in \S_i} z_u & \le & K T &\\
 \sum_{v \in \S_i} z_v \p_{vu} & = & w_u & \forall i,  u \in \S_i(T) \setminus \{\rho_i\} \\
z_u & \le & w_u & \forall u \in \S_i(T), \forall i\\
z_u, w_u & \in & [0,1] & \forall u \in \S_i(T), \forall i\\
\end{array}\]
}

\noindent We use \lpone\ to refer to the value and (\lpone) as the LP.

\begin{claim}
\label{claim1} \label{lem:key}
Let $OPT$ be the value of the optimum policy. Then $OPT \leq \lpone$. 
\end{claim} 
\begin{proof}
  We show that the $\{w^*_u,z^*_u\}$ as defined above, corresponding
  to a globally optimal policy $\P^*$, are feasible for the
  constraints of the LP.

  The first constraint follows by linearity of expectation: The
  expected time for which arm $i$ is played by $\P^*$ is $T_i =
  \sum_{u \in \S_i(T)} z^*_u$. Since at most $K$ arms are played every
  step, we must have $\sum_i T_i \le KT$, which is the first
  constraint. Note that $T_i \le T$, since we have truncated $\S_i$ to
  have only states corresponding to at most $T$ observations.
 
  The second constraint simply encodes that the probability of
  reaching a state $u \in \S_i(T)$ precisely the probability with which
  it is played in some state $v \in \S_i(T)$, times the probability
  $\p_{vu}$ that it reaches $u$ conditioned on that play. The
  constraint $z_u \le w_u$ simply captures that the event an arm is
  played in state $u \in \S_i(T)$ only happens if the arm reaches this
  state.

  The objective is precisely the expected reward of the optimal policy
  -- recall that the reward of playing arm $i$ in state $u$ is $r_u$.
  Hence, the LP is a relaxation of the optimal policy.
\end{proof}

\noindent A similar LP formulation was proposed for the multi-armed bandit
problem by Whittle~\cite{whittle} and Bertsimas and
Nino-Mora~\cite{nino}; however, one key difference is that {\em we
  ignore the time at which the arm is played} in defining the LP variables.

\paragraph{From Linear Programs to Single Arm Policies:}
\label{sec:single} 
The optimal solution to (\lpone) clearly does not
directly correspond to a feasible policy since the variables do not
faithfully capture the joint evolution of the states of different arms
-- in particular, the optimum solution to (\lpone) enforces the horizon
$T$ only in expectation over all decision paths, and not individually
on each decision path. In this paper, we present such relaxations for
each variant we consider, and develop techniques to solve the
relaxation and convert the solution to a feasible policy while
preserving a large fraction of the reward.  Below, we present an
interpretation of any feasible solution to (\lpone), which also helps
us compact representation subsequently.

Let $\langle w_u,x_u,z_u \rangle$ denote any feasible solution to the
(\lpone). Assume w.l.o.g. that $w_{\rho_i} = 1$ for all $i$. Ignoring
the first constraint of (\lpone) for the time being, the remaining
constraints encode a separate policy $\P_i$ for each arm as follows:
Consider any arm $i$ in isolation. The play starts at state
$\rho_i$. The arm is played with probability $z_{\rho_i}$, so that
state $u \in \S_i(T)$ is reached with probability $z_{\rho_i}
\p_{\rho_i u}$. Similarly, conditioned on reaching state $u \in \S_i(T)$,
with probability $z_u/w_u$, arm $i$ is played. This yields a
policy $\P_i$ for arm $i$ which is described in Figure~\ref{fig1}.

\begin{figure*}[htbp]
\framebox{
\begin{minipage}{6.0in}
{\bf Policy $\P_i$:}
If arm $i$ is currently in state $u$, then choose $q \in [0,w_u]$ uniformly at random:
\begin{tabbing}
1.\ \ \= If $q \in [0,z_u]$, then play the arm. \\
2.\> If $q \in (z_u,w_u]$, then stop executing $\P_i$.
\end{tabbing}
\end{minipage}
}
\caption{\label{fig1} The Policy $\P_i$.}
\end{figure*}

For policy $P_i$, it is easy to see by induction that if state $u
\in \S_i(T)$ is reached by the policy with probability $w_u$, then
state $u \in \S_i(T)$ is reached {\em and} arm $i$ is played with
probability $z_u$. Therefore, $P_i$ satisfies all the constraints
of (\lpone) except the first. The first constraint 
simply encodes that the total expected number of plays made by all
these single-arm policies $\{P_i\}$ is at most $KT$.

Note that any feasible solution of (\lpone) defines a collection of
single arm policies $\{\P_i\}$ and each $\P_i \in \C_i(T)$ as
defined in Definition~\ref{singlearmdef}. 
Using the notation in Definition~\ref{notationdef},
$ R(P_i) = \sum_{u \in \S_i(T)} r_u z_u$ and $\T(P_i) = \sum_{u \in \S_i(T)} z_u$ we can represent (\lpone) as follows:

{
\[ \lpone =  \mbox{Max}_{\{P_i \in \C_i(T)\}} \left\{  \sum_i R(P_i) \ |\   \sum_i  \T(P_i)  \leq  KT \right \} \]
}

\subsection{The Idea of Truncation}
\label{part1}
We now show that the time horizon of a single-arm policy on $T$ steps
can be reduced to $\beta T$ for constant $\beta \le 1$ while
sacrificing a constant factor in the reward. We note that though the
statement seems simple, this theorem {\em only} applies to single-arm
policies and is {\em not true} for the global policy executing over
multiple arms.  The proof of this theorem uses the martingale
structure of the rewards, and proceeds via a stopping time argument.
A similar statement is proved in \cite[Lemma 2]{farias} using an
inductive argument; however that exact proof would require
modifications in the varied setting we are applying it and therefore
we restate and prove the theorem.

\begin{definition}
Given a policy $\P$, a {\bf decision path} is a sequence of (observed reward, action) pairs encountered by the policy till stopping. Different decision paths are encountered by the policy with different probabilities, which depend on the observed rewards.
\end{definition}

\noindent Recall the notation $R(\P),\T(\P)$ introduced in Definition~\ref{notationdef}.

\begin{theorem}
\label{statetheorem}(The Truncation Theorem.)
Consider an arbitrary single arm policy $\P$  executing over state space $\S$ defined by $i.i.d.$ draws from a reward distribution $D(\theta)$, where a prior distribution $\D$ is specified over the unknown parameter $\theta$. Suppose $\E[D(\theta)] \ge 0$ for all parameter choices $\theta$. Consider a different policy $\P'$ that executes the decisions of $\P$ but stops after making (at least) $\beta$ fraction of the plays on any decision path  of $\P$. Then (i) $R(\P') \geq \beta R(\P)$ and (ii) $\T(\P') \leq \T(\P)$.
\end{theorem}
\begin{proof}
  We view the system as follows. The unknown reward distribution
  $D(\theta)$ has mean $\mu(\theta)$, and the initial prior $\D$ over
  $\theta$ specifies a prior distribution $f(\mu)$ over possible
  $\mu$.  Consider the execution of $\P$ conditioned on the event that
  the (unknown) mean is $\mu$, and denote this execution as $\P(\mu)$.
  Let $R(\P(\mu)), \T(\P(\mu))$ denote the expected reward and the
  number of plays of the policy $\P$ conditioned on this event. We
  have $R(\P) = \int R(\P(\mu)) f(\mu) d\mu$ and $\T(\P) = \int
  \T(\P(\mu)) f(\mu) d\mu$.

Suppose the execution $\P(\mu)$ reaches some state $v \in \S$ on
decision path $q$ and decides to make a play there. Since the rewards
are $i.i.d.$ draws, the expected reward of the next play at $v$ is
$\mu$ {\em regardless} of the state. We charge the reward to the path
$q$ as follows: Regardless of the actual observation at state $v$
corresponding to path $q$, we charge reward exactly equal to $\mu$ to
$q$ for making a play at state $v$. In other terms, we charge reward
exactly $\mu$ whenever a play is made is any state, regardless of the
actual observation. It is clear that the expected reward of any play
is preserved by this accounting scheme, since the rewards are $i.i.d.$
draws from a distribution with mean $\mu$.

By linearity of expectation, the above charging scheme means that for
any decision path $q$, the expected reward is simply $\mu$ times the
number of plays made on this path. More formally, suppose decision
path $q$ involves $l(q)$ plays and that is taken by $\P(\mu)$ with
probability $g(q)$. Then
$$R(\P(\mu)) = \sum_q \mu \times l(q) \times g(q) \qquad \mbox{and} \qquad \T(\P(\mu)) = \sum_q l(q) \times g(q) $$

For any $\mu$, the policy $\P'(\mu)$ encounters the same distribution
over decision paths as $\P(\mu)$, except that on any decision path,
the execution stops after making at least $\beta$ fraction of the
plays. This means that for any decision path $q$ of $\P(\mu)$, the
execution $\P'(\mu)$ makes at least $\beta l(q)$ plays which yield at
least $\beta \mu l(q)$ expected reward. The theorem now follows by
integrating over all $q$ and $\mu$.
\end{proof}

Given the equivalence of randomized policies and fractional solutions
to the relaxation (\lpone), the
relaxation has a compact representation: Simply find one single-arm
policy for each arm, so that the resulting ensemble of policies have
expected number of plays at most $KT$, and maximum expected reward.
This yields the following equivalent version:

{
\[ \lpone =  \mbox{Max}_{\{P_i \in \C_i(T)\}} \left\{  \sum_i R(P_i) \ |\   \sum_i  \T(P_i)  \leq  KT \right \} \]
}

In subsequent sections, wherever possible, we work with the compact LP
formulation directly. We note that the single-arm policies in these
sections have a richer set of actions, and operate on a richer
state-space; nevertheless, the derivation of the LP relaxation will be
nearly identical to that described in this section. We point out
differences as and when appropriate.

\subsection{A $2$-Approximation using Irrevocable Policies and a Tight Lower Bound}
\label{simple}
\newcommand{\finlp}{\mbox{\sc OPT}}
\newcommand{\G}{\mathcal{G}}
\renewcommand {\p}{\mathbf{p}} 

Recall that the optimum solution of (\lpone), as interpreted in
Figure~\ref{fig1}) will be collection of single arm policies
$\{\P^*_i\}$ such that $\lpone \leq \sum_i R(\P^*_i)$ and $\sum_i
\T(\P^*_i) \le KT$. Such a solution of (\lpone) can be found in
polynomial time. 
Given a collection of single arm policies $\P_i$ which satisfy $\sum_i \T(\P_i) \leq KT$ we introduce the final scheduling 
policy as shown in Figure~\ref{fh}.

\begin{figure*}[ht]
\fbox{
\begin{minipage}{6.0in}
{\small
{\bf The Combined Final Policy}
\begin{enumerate}
\item Solve (LP1) to obtain a collection of single-arm policies $\P_i$.
\item Order the arms in order of $\frac{R(\P_i)}{\T(\P_i)}$. \label{needno}
\item Start with the first $K$ policies in the order specified in Step~\ref{needno}. These policies are inspected; the remaining policies are uninspected.
\begin{enumerate} 
\item If the decision in $\P_i$ is to quit, then move to the first uninspected arm in the order (say $\P_j$) and start executing $\P_j$. This is similar to scheduling $K$ parallel machines. 
\item If the horizon $T$ is reached, the overall policy stops execution. 
Note $KT \geq \sum_i \T(\P_i)$.
\end{enumerate}
\end{enumerate} }
\end{minipage}
}
\caption{The Final Policy for the Finite Horizon MAB Problem\label{fh}}
\end{figure*}

\begin{figure*}[ht]
\begin{centering}
\includegraphics[scale=0.5]{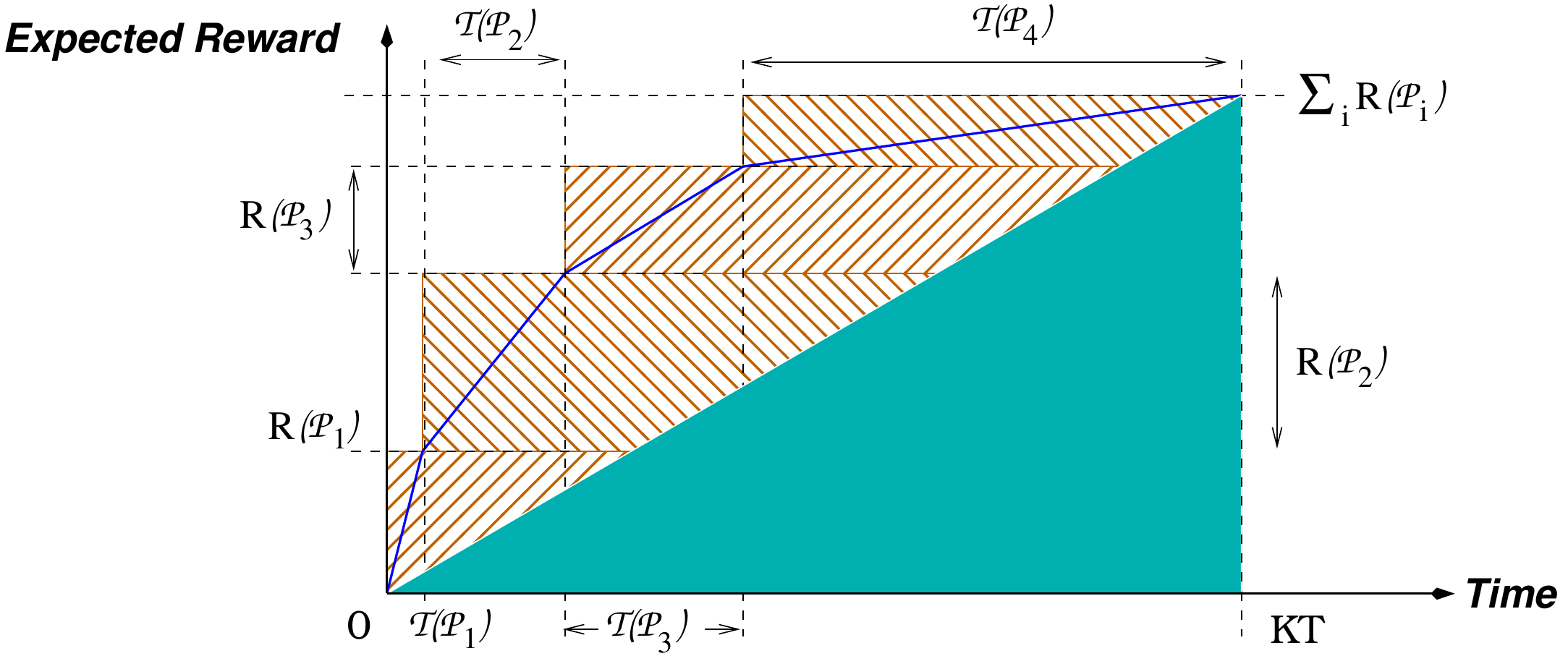}
\caption{\label{concavechainfig} The accounting process of Lemma~\ref{greedyorder} explained pictorially.}
\end{centering}
\end{figure*}

\begin{lemma}
\label{greedyorder}
The policy outlined in Figure~\ref{fh} provides an 
expected reward of at least $\frac12 \sum_i R(\P_i)$.
\end{lemma}
\begin{proof}
Let the number of plays of arm $i$ be $T_i$. We know $\E[T_i] =
\T(\P_i)$ and $\sum_i \T(\P_i) \leq KT$. We start playing arm $i$ after
$\sum_{j < i} T_j$ plays (if the sum is less than $T$); which means 
the remaining horizon for $\P_i$ is $T  - \min\{T, \frac1K \sum_{j<i} T_j \}$.
We apply the Truncation Theorem~\ref{statetheorem} with $\beta=1 - \frac{1}{T}
\min \{ T,\frac1K \sum_{j<i} T_j \}$ and the expected reward of
$\P_i$ continuing from that starting point onward is $ \left (1 - \frac{1}{T}
\min \{ T,\frac1K\sum_{j<i} T_j \} \right)R(\P_i)$. Note that this is a
consequence of the independence of arm $i$ from $\sum_{j<i} T_j$.
Thus the total expected reward $R$ is
{\small
\begin{eqnarray} 
R  & = & \E \left[ \sum_i \left(1 - \frac{1}{T} \min \left \{ T, \frac1K \sum_{j<i} T_j \right\} \right)R(\P_i) \right] \geq  
\E \left[ \sum_i \left(1 - \frac{1}{KT}\sum_{j<i} T_j \right)R(\P_i) \right] \nonumber \\
& = & 
\sum_i \left (1 - \frac{1}{KT}\sum_{j<i} \E[T_j] \right)R(\P_i) =
\sum_i \left (1 - \frac{\sum_{j<i} \T(\P_j)}{KT} \right)R(\P_i) \label{mainequation} 
\end{eqnarray}
}
The bound in Equation~\ref{mainequation} is best represented
pictorially; for example consider Figure~\ref{concavechainfig}.
Equation~\ref{mainequation} indicates that $R$ is $\frac1{KT}$
fraction of the shaded area in Figure~\ref{concavechainfig}, which
contains the triangle of area $\frac12 KT \sum_i R(\P_i)$.
Therefore $R \geq \frac12 \sum_j R(\P_j)$ and the lemma follows.
\end{proof}

\begin{theorem}
\label{unitmab}
There exists a simple $2$-approximation for the finite horizon Bayesian MAB  problem (with budgets and arbitrary priors) using an irrevocable scheduling policy. Such a policy can be found in polynomial time (solving a linear program).
\end{theorem}
\begin{proof}
  Given a collection of single arm policies $\{\P^*_i\}$ which
  correspond to the optimum solution of (\lpone), we apply
  Lemma~\ref{greedyorder} with $\P_i = \P^*_i$; since the collection
  $\{\P^*_i\}$ is feasible, namely, $\sum_i \T(\P^*_i) \leq KT$.
  Therefore the expected reward is at least $\frac12 \sum_i R(\P^*_i)
  \geq \frac12\lpone$ and the theorem follows.
\end{proof}

\paragraph{Tight example of the analysis.}
We show that the gap of the optimum policy and LP1 is a
factor of $2 - O(\frac1n)$, even for unit length plays. Consider the
following situation: We have two ``types'' of arms. The type I arm
gives a reward $0$ with probability $a = 1/n$ and $1$ otherwise. The type
II arm always gives a reward $0$.  We have $n$ independent arms. Each
has an identical prior distribution of being type I with probability
$p = 1/n^2$ and type II otherwise. Set $T = n$.

Consider the symmetric LP solution that allocates one play to each
arm; if it observed a $1$, it plays the arm for $n$ steps. The
expected number of plays made is $n + O(1/n)$, and the expected reward
is $n \times 1/n = 1$. Therefore, $\lpone \ge 1 - O(1/n)$.

Consider the optimum policy. We first observe that if the policy ever
sees a reward $1$ then the optimum policy has found one of the type II
arms, and the policy will continue to play this arm for the rest of
the time horizon. At any point of time before the time horizon, since
$T=n$, there is always at least one arm which has not been played
yet. Suppose the policy plays an arm and observe the reward $0$, then the
posterior probability of this arm being type II increased. So the
optimum policy should not prefer this currently played arm over an
unplayed arm. Thus the optimum policy would be to order the arms
arbitrarily and make a single play on every new arm. If the outcome is
$0$, the policy quits, otherwise the policy keeps playing the arm for
the rest of the horizon. The reward of the optimum policy can thus be
bounded by $\sum_{x=0}^{T-1} ap(1 + (T-x-1)a) = pa^2T(T+1)/2 +
(1-a)/n= \frac12 + O(\frac1n)$. Thus the gap is a factor of $2 -
O(\frac1n)$.

\subsection{Weak Coupling, Efficient Algorithms and Compact Representations}
\label{sec:compact}
\label{sec:lagindex}

We outline how to solve (\lpone) efficiently using a standard
application of weak duality. Recall,

{
\[ \lpone =  \mbox{Max}_{\{P_i \in \C_i(T)\}} \left\{  \sum_i R(P_i) \ |\   \sum_i  \T(P_i)  \leq  KT \right \} \]
}

\noindent We take the Lagrangian of the coupling
constraint to obtain:

{
\[ \mbox{\sc LPLag}(\lambda) = \mbox{Max}_{\{P_i \in \C_i(T)\}} \left\{ K T \lambda +  \sum_i \left( R(P_i) - \lambda \T(P_i) \right) \right\} \]
}

Note now that there are no constraints connecting the arms, so that the optimal policy is obtained by solving {\sc LPLag}$(\lambda)$ separately for each arm.

\begin{definition}
Let $Q_i(\lambda) = \mbox{Max}_{P_i \in \C_i(T)} R(P_i) - \lambda \T(P_i)$ 
denote the optimal solution to $\lplag(\lambda)$ restricted to arm 
$i$, so that $\lplag(\lambda) = KT \lambda+   \sum_i Q_i(\lambda)$. 
Let $\L_i(\lambda)$ denote the corresponding optimal policy for arm $i$.
As a convention if $Q_i(\lambda)=0$ then we choose $\L_i(\lambda)$ to be the trivial policy which does nothing.
 \end{definition}

\begin{lemma}
\label{weakduality}
For any $\lambda\geq 0$ we have $\lplag(\lambda) = \lambda KT + \sum_i Q_i(\lambda) \geq \lpone \geq OPT$.
\end{lemma}

\noindent The above Lemma is an easy consequence of weak duality.
We now compute $Q_i(\lambda)$.
 
\begin{lemma}
\label{mainlem}
$Q_i(\lambda)$ and the corresponding single-arm policy $\L_i(\lambda)$
(completely specifying an optimum solution of {\sc LPLag}$(\lambda)$)
can be computed in time $O(\edge(\S_i))$. 
For $\lambda \geq 0$, $R(\L_i(\lambda)), \T(\L_i(\lambda))$ 
and $Q_i(\lambda)$ are non-increasing as $\lambda$ increases.
\end{lemma}
\begin{proof}
  We use a straightforward bottom up dynamic program over the DAG
  represented by $\S_i(T)$ which is the statespace $\S_i$ restricted
  to a horizon $T$.

  Let Gain$(u)$ to be the maximum of the objective of the single-arm
  policy conditioned on starting at $u \in \S_i(T)$. If $u$ has no
  children, then if we ``play'' at node $u$, then Gain$(u)= r_u
  -\lambda$ in this case.  Stopping corresponds
  to Gain$(u)=0$.  Therefore we set Gain$(u)=\max \{ 0, r _u -\lambda
  \}$ in this case.
If $u$ had children, playing corresponds to a gain of $r_u -\lambda + \sum_v \p_{uv} \mbox{Gain}(v)$.  Therefore:

{
 \[ \mbox{Gain}(u) = \mbox{Max} \left\{ 0, \ \  r_u -\lambda  + \sum_v \p_{uv} \mbox{Gain}(v) \right\} \] 
}

The policy $P_u$ constructed by the dynamic program rooted
at $u \in S_i(T)$ satisfies the invariant $R(P_u)=\mbox{Gain}(u) +
\lambda \T(P_u)$. This immediately implies Gain$(\rho_i)=Q_i(\lambda)$. 
Note that we can ensure that Gain$(u)=0$ is obtained by the trivial policy at $u$ of doing nothing.

Moreover the decision to play at $\lambda'$ also implies a decision to
play at $\lambda < \lambda'$. This trivially implies that
$\T(\L_i(\lambda)),R(\L_i(\lambda))$ is non-increasing as $\lambda$
increases.  Likewise, observe that for every $u \in \S_i(T)$, we have
$\mbox{Gain}(u)$ is nonincreasing as $\lambda$ increases, and
therefore $Q_i(\lambda)$ is nonincreasing.
\end{proof}
\noindent In terms of the variables of the original (\lpone), $Q_i(\lambda)$ is defined as follows:
\begin{equation}
\begin{array}{rcll}
\displaystyle Q_i(\lambda) & = & \displaystyle \mbox{Max} \sum_{i=1}^n \sum_{u \in \S_i(T)} (r_u - \lambda )z_u & \\ 
 \sum_{v \in \S_i} z_v \p_{vu} & = & w_u & \forall i,  u \in \S_i(T) \setminus \{\rho_i\} \\
z_u & \le & w_u & \forall u \in \S_i(T), \forall i\\
z_u, w_u & \in & [0,1] & \forall u \in \S_i(T), \forall i\\
\end{array} \label{qidef}
\end{equation}
\noindent The solution presented in Lemma~\ref{mainlem} shows that the
optimum policy $\L_i(\lambda)$ satisfies $z_u=0$ (no play, or
\mbox{Gain}$(u)=0$) or $z_u=w_u$ (play or\mbox{Gain}$(u)>0$); which
are to expected using complementary slackness \cite{book}. 
By a standard application of weak duality (see for
instance~\cite{JainV99}), a $(1+\epsilon)$ approximate solution to
$(LP1)$ can be obtained by taking a convex combination of the
solutions to {\sc LPLag}$(\lambda)$ for two values $\lambda^-$ and
$\lambda^+$; these can be computed by binary
search.  This yields the following.

\begin{theorem} 
\label{blahtheorem} 
In time $O((\sum_i \edge(\S_i))\log (nT/\epsilon))$, we can compute 
quantities $a,\lambda^-$ and
$\lambda^+$, where $|\lambda^- - \lambda^+| \leq \epsilon OPT/(KT)$
and a fraction $a \in [0,1)$ so that if $\P^*_i$ denotes the
single-arm policy that executes $\L_i(\lambda^-)$ with probability
$a$ and $\L_i(\lambda^+)$ with probability $(1-a)$, then
these policies are feasible for (\lpone) with objective at least
$OPT/(1+\epsilon)$.  
\end{theorem}
\begin{proof}
  Observe that for $\lambda=0$, if we satisfy the constraint $\sum_i
  \T(\L_i(\lambda)) \leq KT$ then the theorem is immediately true
  based on Lemmas ~\ref{weak1} and \ref{mainlem} (setting
  $\epsilon=0$). So in the remainder we assume that $\sum_i
  \T(\L_i(0)) > KT$.
  Note that $OPT \leq \sum_i Q_i(0)$. Moreover, for all $i$, $OPT \geq
  Q_i(0)$ since the optimum can disregard all other arms.  Let
  $M=\sum_i Q_i(0)$ and thus $M \leq n OPT$.

\smallskip
  Now if we set $\lambda=2M$ then all
  $Q_i(\lambda)=0$ because the penalty of $\lambda$ to the root node
  is larger than the total reward of the policy.  Thus all $\L_i(M)$
  are the trivial null policy.  In this case $\sum_i \T(\L_i(M))
  = 0 < KT$.

Therefore we can maintain the interval defined by the two numbers 
$\lambda^-<\lambda^+$ such that $\sum_i \T(\L_i(\lambda^-)) > KT$ and $\sum_i \T(\L_i(\lambda^+)) \leq KT$. Initially $\lambda^-=0$ and $\lambda^+=\infty$. We can now perform a binary search and maintain the properties till we have $\lambda^+ - \lambda^- \leq \epsilon M/(2nKT) \leq \epsilon OPT/(2KT)$.
Since $\sum_i \T(\L_i(\lambda^-)) > KT \geq \sum_i \T(\L_i(\lambda^+)) $ there exists an unique $a \in [0,1)$ such that 
\[ a \left( \T(\L_i(\lambda^-)) \right) + (1-a) \left( \T(\L_i(\lambda^+)) \right) = KT \]
Note that for such an $a$, we have $ \sum_i \T(\P_i^*) = a \left( T(\L_i(\lambda^-)) \right) + (1-a) \left( \T(\L_i(\lambda^+)) \right) = KT$
thereby satisfying the main constraint in the compact representation of \lpone. Observe that for $\lambda \in \{ \lambda^-,\lambda^+\}$, 
\begin{equation}
\label{needref001}
\lambda KT + \sum_i \left\{ R(\L_i(\lambda) - \lambda \T(\L_i(\lambda)) \right \} \geq OPT
\end{equation}
 using the definition of $Q_i(\lambda),\L_i(\lambda)$ and Lemma~\ref{weak1}. As a consequence, since $R(\P_i^*)=a R(\L_i(\lambda^-) + (1-a) R(\L_i(\lambda^+)$; we have:
{
\begin{align*}
OPT & \leq a \left[ \lambda^- KT + \sum_i \left\{ R(\L_i(\lambda^-) - \lambda^- \T(\L_i(\lambda^-)) \right\} \right ] + \\
& \hspace{1.5in}  (1-a) \left[ \lambda^+ KT + \sum_i \left\{ R(\L_i(\lambda^+) - \lambda^+ \T(\L_i(\lambda^+)) \right\} \right ] \\
& = (a \lambda^- + (1-a) \lambda^+)KT + \sum_i R(\P_i^*) - \lambda^-  \sum_i \T(\P_i^*) - (1-a) (\lambda^+ - \lambda^-)  \sum_i \T(\L_i(\lambda^+)) 
\end{align*}
}
\noindent and since $\sum_i \T(\P^*_i) =KT$ the last equation rewrites to
\[ OPT \leq \sum_i R(\P_i^*) + (1-a) (\lambda^+ - \lambda^-)  \left[ KT - \sum_i \T(\L_i(\lambda^+)) \right] \]
\noindent 
But that implies $OPT \leq \sum_i R(\P_i^*) + (1-a) (\lambda^+ -
\lambda^-) KT \leq \sum_i R(\P_i^*) + \epsilon OPT/2$. Since for $x \in (0,1]$ we have $\frac1{1+x} \leq 1-x/2$ we have $OPT/(1+\epsilon) \leq \sum_i R(\P_i^*)$.

Observe that the initial size of the interval is $2M$ and the final
size is at most $\epsilon M/(nKT)$. Therefore the number of binary
searches is $\log \frac{nKT}{\epsilon} = O(\log (nT/\epsilon))$ since
$k<n$. The theorem follows.
\end{proof}

\noindent Using Theorem~\ref{blahtheorem}, Lemma~\ref{greedyorder} and a rescaling of $\epsilon$ the following is immediate:

\begin{corollary}\label{blah2}
Given any $\epsilon \in (0,1]$,
  using $O( (\sum_i \edge(\S_i) )\log (nT/\epsilon))$ time 
  we can compute a $(2+\epsilon)$-approximation to the finite horizon Bayesian MAB using irrevocable policies.
\end{corollary}

\begin{figure*}[ht]
\begin{centering}
\includegraphics[scale=0.5]{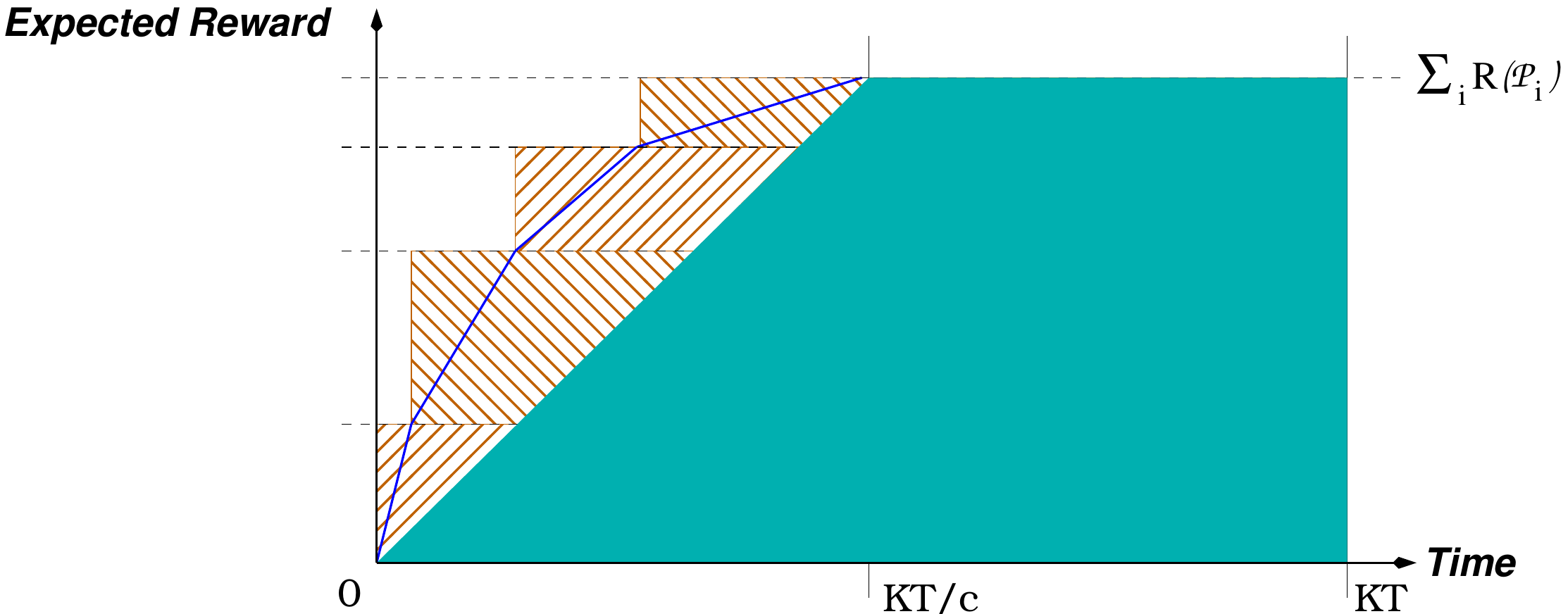}
\caption{\label{concavechainfig2} The accounting process of Corollary~\ref{approxblahtheorem} explained pictorially.}
\end{centering}
\end{figure*}

\newcommand{\CC}{{\mathfrak C}}
\paragraph{Further applications of Theorem~\ref{blahtheorem}:} We prove a corollary which will be useful later;
\begin{corollary}
\label{approxblahtheorem}
If in time $O(\tau)$ we can compute an $c$-approximation to
$\sum_i Q_i(\lambda,\CC)$ where $Q_i(\lambda,\CC)$ is
$Q_i(\lambda)$ is defined by the maximizing system~(\ref{qidef}) with the additional constraint $\CC$ over single the arm policy; 
then we can compute a
$(c+\epsilon)$-approximation for (\lpone) for any $\epsilon \in
(0,1]$, which satisfies $\sum_i \T(\P_i) = KT/c$ and the additional
restriction over the same constraints $\CC$ over the single arm
policies in time $O(\tau \log (nT/\epsilon))$. Further 
the scheduling policy in Figure~\ref{fh} now provides a $\frac{2c(c+\epsilon)}{2c-1}$ approximation to $OPT(\CC)$ which is the optimum solution which obeys the constraint
  that $\sum_i \T(\P_i) \leq KT$ along with the single arm constraints
  $\CC$.
\end{corollary}
\begin{proof} 
  A relaxation of $OPT(\CC)$ can be expressed as a mathematical
  program ($\CC$ need not be linear constraints) where we have
  (\lpone) with additional constraints $\CC$. However weak duality
  still holds and for any $\lambda \geq 0$,
\begin{equation}
\label{needref002}
 \lambda KT + \sum_i Q_i(\lambda,\CC) \geq OPT(\CC)
\end{equation}
Now the proof of the corollary follows from replicating the proof of
Theorem~\ref{blahtheorem} with the following consequence of the
approximation algorithm, instead of Equation~(\ref{needref001}), 
\begin{equation}
\label{needref003}
\lambda KT + c \left(\sum_{i: \L_i(\lambda) \mbox{ obeys } \CC} \left\{ R(\L_i(\lambda) - \lambda \T(\L_i(\lambda)) \right \} \right) \geq OPT(\CC)
\end{equation} 
which follows from equation~\ref{needref002} and the $c$-
approximation of $Q_i(\lambda,\CC)$.
Observe that
now if we choose $\lambda^-,\lambda^+$ such that $\sum_i
\T(\L_i(\lambda^-)) > KT/c \geq \sum_i \T(\L_i(\lambda^+))$ then we can ensure that $\sum_i \T(\P^*_i) = KT/c$. We still execute the policies as described in
Figure~\ref{fh}, and the expected reward is at least (following the identical logic as in Lemma~\ref{greedyorder}):
\[ R' \geq \sum_i \left (1 - \frac{\sum_{j<i} \T(\P_j)}{KT}
\right)R(\P_i) \] 
Pictorially, $R'$ is at least $\frac1{KT}$ times the
area of the shaded triangle and the rectangle as shown in Figure~\ref{concavechainfig2}, which amounts to
$(1-\frac1{2c}) \sum_i R(\P_i)$. 
Observe that $c=1$ corresponds to the statement of
Lemma~\ref{greedyorder}. 
The overall approximation is therefore $\frac{2c(c+\epsilon)}{2c-1}$ which proves the theorem.
\end{proof}

\subsection{Comparison to Gittins Index}
\label{sec:gittins}
The most widely used policy for the discounted version of the
multi-armed bandit problem is the {\em Gittins index}
policy~\cite{GJ74}.  Recall the single-arm policies constructed in
Section~\ref{sec:lagindex}. For arm $i$, consider the policy
$\L_i(\lambda)$ corresponding to the value $Q_i(\lambda)$.  We can
account for the reward of this policy as follows. Suppose any policy
is given fixed reward $\lambda$ per play (so that if the expected
number of plays is $\T(P)$, the policy earns $\lambda \T(P)$). Then,
the value $Q_i(\lambda)$ is the optimal {\em excess} reward a policy
can earn given these average values, since this is precisely $\max_{P
  \in \C_i} (R(P) - \lambda \T(P))$. To solve our relaxation (\lpone), we
find $\lambda$ so that the total expected number of plays made by the
single-arm policies for different $i$ sums to $T$.
The definition of $Q_i(\lambda)$ can be generalized to an equivalent
definition $Q_i(\lambda,u)$ for policies whose start state is $u \in
\S_i(T)$ (instead of being the root $\rho_i$). The Gittins index for $u
\in \S_i(T)$ can be defined as:
$$\mbox{Gittins Index} = \Pi_i(u) = \max \{\lambda \ |\  Q_i(\lambda,u) > 0\}$$
In other words, the Gittins index for state $u$ is the maximum value
of $R(P)/\T(P)$ over policies restricted to starting at state $u$ (and
making at least one play), {\em i.e.}, the maximum amortized per-step
long term reward obtainable by playing at $u$. The Gittins index
policy works as follows: At any time step, play the arm $i$ whose
current state $u$ has largest index $\Pi_i(u)$. For the discounted
reward version of the problem, such an index policy yields the optimal
solution. This is not true for the finite horizon version, and the
other variants we consider below. Nevertheless, the starting point for
our algorithms is the solution to (\lpone), and as shown above, this has
computational complexity similar to the computation of the Gittins
index.

In contrast to the Gittins index, our policies are based on computing
(as in Theorem~\ref{blahtheorem}) {\em one} global penalty $\lambda^*$
across all arms by solving (LP1); consider the case $\lambda^+\approx\lambda^*\approx\lambda^-$ in Theorem~\ref{blahtheorem}. For this penalty, for each arm $i$
and state $u \in \S_i(T)$, the policy $\L_i(\lambda^*)$ makes a decision
on whether to play or not play. We execute these decisions, and impose
a fixed priority over arms to break ties in case multiple arms decide
to play.

\newcommand{\Q}{\mathcal{Q}}
\section{Traversal Dependent Bayesian MAB Problems}
\label{sec:metric}

In this section we consider how the constraints on a traversal of
different bandit arms can affect the approximation algorithm. A
concrete example of such traversal related constraint is the {\bf
  Bayesian MAB Problem with Switching Costs} where there is a {\em cost}
of switching between arms. Denote the cost of switching from arm $i$
to arm $j$ as $\ell_{ij} \in \mathcal{Z}^+$. The system starts at an
arm $i_0$. The goal is to maximize the expected reward subject to
{\bf rigid} constraints that (i) the total number of plays is at most
$T$ and (ii) the total switching cost is at most $L$ on all decision
paths. This problem has received significant attention, see the discussion in 
Section~\ref{sec:intro} and in \cite{BS94} - however efficient solutions 
with provable bounds in the Bayesian setting has been elusive.

A classic example of such a switching cost problem can be when the
costs $\ell_{ij}$ define a distance {\bf metric}, which is natural in
most navigational settings and was considered earlier in \cite{GM09}.
Here we will provide a $X$ approximation for that problem improving
the $12$-approximation provided in \cite{GM09}. 

However a strong motivation of this section is to continue developing
general techniques for Bayesian MAB problems and therefore we take a
slightly indirect route. We first consider a different problem: {\bf
  Finite Horizon Bayesian MAB using Arbitrary Order Irrevocable
  Policies}. In this problem, once we have decided upon the {\em set}
of arms to play then an adversary provides us with a specific order
such that if we start playing an arm $i$ then we cannot visit/revisit
any arm before arm $i$ in that said order\footnote{This admits $K>1$ cases. The adversary does not have the knowledge of the true reward values.}.  We will
provide an efficient (again near linear time in the input sparsity)
$Y$-approximation for this problem. Since switching costs are
often used to model economic interactions in the Bandit setting, as in
\cite{BS94,GGW,labor1,labor2}, the adversarially ordered traversal
problem is an interesting subproblem in its own right. In
addition, there are two key benefits of this approach.
\begin{itemize}
\item First, the analysis
technique for arbitrary (or adversarial) order will disentangle the
decisions between the constraint associated with the traversal and the constraint associated with the finite horizon. Note that all these traversal problems encode a natural
combinatorial optimization problem and are often MAX SNP
Hard --- and this disentanglement
and isolation of the combinatorial difficulty produces natural policies.
\item Second, because we use irrevocable policies, the switching cost is only relevant for the algorithm for the first transition from $i$ to $j$. The proof presented herein will
remain exactly the same if the second transition of $i$ to $j$ costs
more or less than the first transition, as long as the costs are
positive! 
\end{itemize}

\noindent{\bf Roadmap:}
We first discuss the Arbitrary Order Irrevocable
Policy problem in Section~\ref{arbitrary}. We then show in
Section~\ref{metric:app} how the analysis applies to the Bayesian MAB
problem with Metric Switching costs; the analysis will extend beyond
the metric assumption as long as a certain traversal type problem
({\em Orienteering Problem}) can be approximated to small factors in
polynomial time.

\subsection{Arbitrary Order Irrevocable Policies for the Bayesian MAB Problem}
\label{arbitrary}
The set up of this problem is similar to the Finite Horizon MAB
problem using irrevocable policies discussed in Section~\ref{finite}.
The only difference is that we cannot use the ordering of
$R(\P_i)/T(P_i)$ as described in the scheduling policy in
Figure~\ref{fh} -- instead we will have to use an arbitrary order {\em
  after the arms and the corresponding policies $\P_i$ have been
  chosen}, that is, the Step~\ref{needno} is not performed.  Note
that the bound of \lpone\ remains a valid upper bound of this problem
-- but Lemma~\ref{greedyorder} does not apply explicitly and
Theorem~\ref{blahtheorem} is not useful implicitly. In what follows,
we prove Theorem~\ref{arbitraryorder} which
replaces them.  The notation used will be the same as in
Section~\ref{finite}.

\begin{figure*}[htbp]
\fbox{
\begin{minipage}{6.0in}
{\small
{\bf The Final Adversarial Order Irrevocable Policy}
\begin{enumerate}
\item Solve (LP1) to obtain a collection of single-arm policies $\P_i$.
\label{firststep2}
\item For each $i$, create policy $\P'_i$ which chooses to plays $\P_i$ with probability $\alpha \in (0,1]$ and with probability $(1-\alpha)$ is the null policy.
\item An adversary orders the arms determining the order in they have to be played. 
\label{needno2}
\item Start with the first $K$ policies in the order specified in Step~\ref{needno2}. 
These policies are inspected; the remaining policies are uninspected.
\begin{enumerate} 
\item If the decision in $\P'_i$ is to quit, then move to the first uninspected arm in the order (say $\P'_j$) and start executing $\P'_j$. This is similar to scheduling $K$ parallel machines. 
\item If the horizon $T$ is reached, the overall policy stops execution. 
Note $KT \geq \sum_i \T(\P'_i)$.
\end{enumerate}
\end{enumerate} }
\end{minipage}
}
\caption{Scheduling where an adversary to decides the traversal order\label{fh2}}
\end{figure*}

\begin{theorem}
\label{arbitraryorder}
The Finite Horizon Bayesian MAB Problem using arbitrary order
irrevocable policies has a $4$ approximation that can be found in
polynomial time and a $(4+\epsilon)$-approximation in $O( (\sum_i
\edge(\S_i) )\log (nT/\epsilon))$ time using the scheduling policy in
Figure~\ref{fh2}.
\end{theorem}
\begin{proof}
  Using the exact same initial arguments as in Lemma~\ref{greedyorder}
  we arrive at the inequality~\ref{mainequation} which states that the
  expected reward $R'$ in this case satisfies:
  \[ R' \geq \sum_i \left( 1 -\frac{\sum_{j<i} \T(\P'_j)}{KT} \right) R(\P'_i) \geq
\left( 1 -\frac{\sum_{j} \T(\P'_j)}{KT} \right) \sum_i R(\P'_i) \]
  Now observe that $R(\P'_i)=\alpha R(\P_i)$ and $\T(\P'_i)=\alpha \T(\P'_i)$ and therefore $R' \geq \alpha(1-\alpha) \sum_i R(\P_i)$. Using an exact solution or Theorem~\ref{blahtheorem}, and setting $\alpha=\frac12$ the statement in this theorem follows.
\end{proof}

\noindent However note that there is a slack in the above analysis
because in Step~\ref{firststep2} in the policy in Figure~\ref{fh2}, we
could have found weakly coupled policies that take a combined horizon
of $2KT$. Balancing this slack will provide us with an alternate
optimization which is useful if we cannot compute $Q_i(\lambda)$ or
(\lpone) in a near optimally fashion. One such example is the Finite
Horizon Bayesian MAB Problem with Metric Switching Costs, which we
discuss next.

\newcommand{\setS}{{\mathfrak S}}

\subsection{Bayesian MAB Problem with Metric Switching Costs}
\label{metric:app}
For simplicity we also assume $K=1$ in this section, that is, the
system is allowed to play only one arm at a time and observe only the
outcome of that played arm. We discuss the $K>1$ at the end of the subsection.
The system starts at an arm $i_0$. A policy, given the outcomes of the
actions so far (which decides the current states of all the arms),
makes one of the following decisions (i) play the arm it is currently
on; (ii) play a different arm (paying the distance cost to switch to
that arm); or (iii) stop. Just as before, a policy obtains reward
$r_u$ if it plays arm $i$ in state $u \in \S_i$. Any policy is also
subject to {\bf rigid} constraints that the total number of plays is
at most $T$ and the total distance cost is at most $L$ on all decision
paths. To begin with, we delete all arms $j$ such that $\ell_{i_0j} > L$. No feasible policy can reach such an arm without exceeding the distance cost budget. 
Let $OPT$ denote both the optimal solution as well as its
expected reward.

\subsubsection{A (Strongly Coupled) Relaxation and its Lagrangian}
We describe a sequence of relaxations to the optimal policy,
culminating with a weakly coupled relaxation. A priori, it is not
clear how to construct such a relaxation, since the switching cost
constraint couples the arms together in an intricate fashion. We
achieve weak coupling via the Lagrangian of a natural LP relaxation,
which we show can be solved as a combinatorial problem called
orienteering over single-arm policies.  
\eat{This section also introduces a
new technique of accounting for the reward via the Lagrange
multiplier, which significantly improves the constant factor; we
present another application in Section~\ref{sec:budget}.}

\begin{definition}
  Let $\C(L,T)$ denote the set of policies on all the remaining arms,
  over a time horizon $T$, that can perform one of two actions: (1)
  Play current arm; or (2) Switch to different arm.  Such policies
  have no constraints on the total number of plays, but are required
  to have distance cost $L$ on all decision paths. Observe that if the
  constraint corresponding to the distance constraint is removed then
  $P \in \C(L,T)$ will decompose to $\{P_i \in \C_i(T)\}$, that is,
  the single arm policies $P_i$ (which are the projections of $P$)
  have at most a horizon of $T$ (See Definition~\ref{singlearmdef} for the definition of $\C_i(T)$).
\end{definition}

Given a policy $P \in \C(L,T)$ define the following quantities in
expectation over the decision paths: Let $R(P)$ be the expected reward
obtained by the policy and let $\T(P)$ denote the expected number of
plays made.  Note that any policy $P \in \C(L,T)$ needs to have distance
cost at most $L$ on {\em all} decision paths. Consider the following
optimization problem, which is still strongly coupled since $\C(L,T)$ is
the space of policies over all the arms, and not the space of
single-arm policies:

 $$  (M1): \qquad \mbox{Max}_{P \in \C(L,T)} \left\{ R(P) \left| \ \T(P)  \le T \right. \right\}   $$ 

\begin{proposition}
$OPT$ is feasible for $(M1)$.
\end{proposition}
\begin{proof}
We have $\T(OPT) \le T$; since $OPT \in \C(L,T)$,  this shows it is feasible for $(M1)$.
\end{proof}

Let the optimum solution of $(M1)$ be $OPT'$ and the corresponding policy be $P^*$ such that $R(P^*)=OPT' \ge OPT$.  Note that $P^*$ need not be feasible for the original problem, since  it enforces the time horizon $T$ only in expectation over the decision paths. We now consider the Lagrangian of the above for $\lambda>0$, and define the problem $M2(\lambda)$:
\begin{definition} 
Let
$V(\lambda) = \mbox{\rm Max}_{P \in \C(L,T)} \left(R(P) - \lambda \T(P) \right)$. 
Let $M2(\lambda)$ be $\mbox{\rm Max}_{P \in \C(L,T)} f_{\lambda}(P)$. 
For a policy 
$P \in \C(T)$ let $f_{\lambda}(P) = \lambda T +  R(P) - \lambda \T(P)$. 
Then,
\[ M2(\lambda): \qquad  \mbox{\rm Max}_{P \in \C(L,T)} \ \  f_{\lambda}(P) = \lambda T + \mbox{\rm Max}_{P \in \C(L,T)} \left(R(P) - \lambda \T(P)  \right) = \lambda T + V(\lambda) \]
\end{definition}

We first relate $OPT$ to the optimal value $\lambda T + V(\lambda)$ of the problem  $M2(\lambda)$.

\begin{lemma}
\label{lem:opti}
For any $\lambda \ge 0$, we have $M2(\lambda) =  \lambda T + V(\lambda) \ge  OPT$.
\end{lemma}
\begin{proof}
This is simply weak duality: For the optimal policy $P^*$ to $(M1)$, we have $\T(P^*) \le T$. Since this policy is feasible for $M2(\lambda)$ for any $\lambda \ge 0$, the claim follows.
\end{proof}

In the Lagrangian formulation, if arm $i$ is played in state $u \in \S_i$, the expected reward obtained is $r_u - \lambda$. We re-iterate that the only constraint on the set of policies $\C$ is that the distance cost is at most $L$ on all decision paths. 

\subsubsection{Structure of $M2(\lambda)$}
The critical insight, which explicitly uses the fact that in the MAB the state of an inactive arm does not change and which allows weak coupling, is the following:

\begin{lemma}
\label{prop:cool}
For any $\lambda \ge 0$, given any $P \in \C$, there exists a $P' \in \C$ that never revisits an arm that it has already played and switched out of, such that $f_{\lambda}(P') \ge f_{\lambda}(P)$.
\end{lemma}
\begin{proof}
  We will use the fact that $\S_i$ is finite in our proof. Suppose $P  \in \C$ revisits an arm. Consider the last point in time, denote this $\alpha$, where the policy is at some arm $i$, makes a  decision to switch to arm $j$ and at some subsequent point in time on some decision path, revisits arm $i$.  Note that after time $\alpha$, the subsequent decision policy $P_{\alpha}$ satisfies the condition that no arm is revisited. Therefore, all plays of arm $j$ are contiguous in time within the policy $P_{\alpha}$. Let $\tilde{P}$ denote the decision policy subsequent to exiting arm $j$; this policy can be different depending on the outcomes of the plays for arm $j$. However, since any such $\tilde{P}$ never revisits arm $j$, we can simply choose that policy $\tilde{P^*}$ with maximum $f_{\lambda}(\tilde{P})$ and execute it regardless of the outcomes of the plays of arm $j$. This yields a new policy $P'$ so that $f_{\lambda}(P') \ge f_{\lambda}(P)$, and furthermore, the distance cost of $P'$ is at most $L$ in all decision paths, so that $P'$ is feasible. By the repeated application of this procedure, the decision policy $P_{\alpha}$ can be changed to the following form without decreasing $f_{\lambda}(P)$: The new policy makes an adaptive set of plays arm $j$; regardless of the outcomes of these plays, the policy switches to the same next arm $k$, again makes an adaptive set of plays; transitions to the same next arm $l$ regardless of the outcome of the plays, and so on, without ever revisiting an arm. We will refer to such a policy as a ``path'' over arms.
   
Recall that the overall decision policy is at arm $i$ just before time $\alpha$. Suppose this arm is revisited in the path $P_{\alpha}$, so that an adaptive set of plays is made for this arm. We modify the policy to make these plays at time $\alpha$, and then switch to arm $j$ regardless of the outcome of these plays. Suppose the original path $P_{\alpha}$ visited arm $i$ from arm $k$ and subsequently switched from $i$ to $l$ regardless of the outcomes of the plays of $i$, the new policy switches from arm $k$ directly to arm $l$.  Since the states of arms that are not played never changes, this movement preserves the states of all the arms, and the moved plays for arm $i$ are stochastically identical to the original plays. Further, the distance cost of the new policy is only smaller, since the cost of switching into and out of arm $i$ is removed. This eliminates $\alpha$ as the last time when the policy is at some arm which is subsequently revisited. By repeated application of the above procedure, the lemma follows.
\end{proof}

Note that the above is {\em not} true for policies restricted to be
feasible for $(M1)$. This is because the step where we use policy
$\tilde{P}$ regardless of the outcome of the plays for arm $j$ need
not preserve the constraint $\T(P) \le T$, since this depends on the
number of plays made for arm $j$.  The Lagrangian $M2(\lambda)$ makes
the overall objective additive in the (new) objective values for each
arm, with the only constraint being that the distance cost is
preserved. Since this cost is preserved in each decision branch, it is
preserved by using the best policy $\tilde{P}$ regardless of the
outcome of the plays for $j$.

\subsubsection{Orienteering and the Final Algorithm}

We now show that the optimal solution to M2$(\lambda)$ is a collection
of single-arm policies connected via a combinatorial optimization
problem termed {\em orienteering}.

\begin{definition}
\label{def:orient}
In the {\em orienteering} problem~\cite{BlumCKLMM03,BansalBCM04,CKP},
we are given a metric space $G(V,E)$, where each node $v \in V$ has a
reward $o_v$. There is a start node $s \in V$, and a distance bound
$L$. The goal is to find that tour $P$ starting at $s$, such that
$\sum_{v \in P} o_v$ is maximized, subject to the length of the tour
being at most $L$. Observe that for any given any $\varepsilon \in (0,1]$
we can discretize $\{o_v| o_v>0\}$ such that within a $(1+\varepsilon)$ approximation we can assume that $o_v$ are integers which are at most $O(n/\varepsilon)$.
\end{definition}

\noindent An immediate consequence of Lemma~\ref{prop:cool} is the
following:

\begin{corollary}
\label{prop:next}
Define a graph $G(V,E)$, where node $i \in V$ corresponds to arm $i$.
The distance between nodes $i$ and $j$ is $\ell_{ij}$, and the reward
of node $i$ is $o_i = Q_i(\lambda)$.  The optimum solution
$V(\lambda)$ of $M2(\lambda)$ is the optimal solution to the
orienteering problem on $G$ starting at node $i_0$ and respecting
rigid distance budget $L$.
\end{corollary}
\begin{proof}
  Consider any $n$-arm policy $P \in \C$. By Lemma~\ref{prop:cool},
  the decision tree of the policy can be morphed into a sequence of
  ``super-nodes'', one for playing each arm, such that the decision
  about which arm to play next is independent of the outcomes for the
  current arm. The policy maximizing $f_{\lambda}(P)$ will therefore
  choose the best policy in $\C_i$ for each single arm as the
  ``super-node''(obtaining objective value precisely $Q_i(\lambda)$),
  and visit these subject to the constraint that the distance cost is
  at most $L$. This is precisely the orienteering problem on the graph
  defined above.
\end{proof}

\begin{theorem}
\label{quote}
For any $\epsilon \in (0,1]$ the orienteering problem has a
$(2+\epsilon)$-approximation that can be found in polynomial time
\cite{CKP}. The authors of \cite{BlumCKLMM03} 
showed that any $c$-approximation for $K=1$ case (where we
choose a single tour) extends to an $(c+1)$-approximation for the 
$K>1$ case where we choose $K$ tours\footnote{They also provided a $4$ approximation for $K=1$, along with bounds for a variety of other traversal problems.}.
\end{theorem}

\noindent We are now ready to present the main theorem for this application,

\begin{theorem}
\label{thm:future2}
For the finite-horizon multi-armed bandit problem with metric
switching costs, for $K=1$ play at a time step 
there exists a polynomial time computable ordering of
the arms and a policy for each arm, such that a solution which plays
the arms using those fixed policies, in that fixed order without
revisiting any arm, has reward at least $1/(4+\varepsilon)$
times that of the best adaptive policy, for any $\varepsilon \in (0,1]$.
For $K>1$ the the reward is at least $1/(4.5 + \varepsilon)$.
\end{theorem}
\begin{proof}
  We will use Corollary~\ref{approxblahtheorem} and use the scheduling
  policy in Figure~\ref{fh2} setting $\alpha=1$. Note that using
  Lemma~\ref{prop:next}, the constraints $\CC$ are null! The traversal
  constraint does not imply any constraints on the internals of a
  single arm policy. However $OPT(\CC)$ still has a rigid constraint of   a traversal cost of at most $L$. 

  Based on the $c(K)$ approximation algorithm and Corollary~\ref{approxblahtheorem} we have a collection of policies $\{\P_i\}$
  such that (i) $\sum_i \T(\P_i) = KT/c(K)$ and (ii) 
  $c(K) \sum_i R(\P_i) \geq
  OPT(K)$ where $c(K)$ is the approximation factor determined by
  Theorem~\ref{quote} and $OPT(K)$ is the corresponding optimum
  solution (again under the rigid traversal cost $L$). Now observe that the reward $R''$ we obtain satisfies:
  \[ R'' \geq \left( 1 -\frac{\sum_{j} \T(\P'_j)}{KT} \right) \sum_i
  R(\P'_i) \geq \left( 1 -\frac{1}{c(K)} \right) \frac1{c(K)} OPT(K)
  \] 
  The theorem follows using $c(1)=2+\epsilon$ and $c(K)=3+\epsilon$
  where $\epsilon=\varepsilon/12$. 
  Note that if $c(K) \leq 2$ then we should have used 
$\alpha=\min\{1,c(K)/2\}$.
\end{proof}

\paragraph{Remark:} Note that the above technique is powerful enough
to approximate any switching cost problem as long as the basic
combinatorial problem is approximable. We note that there exists
approximations for asymmetric distances (directed graphs with
triangle inequality) and other traversal problems in
\cite{BlumCKLMM03}. Modifications to Theorem~\ref{thm:future2} 
will provide a solution for all such problems.

\section{Multi-armed Bandits with Delayed Feedback}
\label{sec:delay}
In this variant of the MAB problem, if an arm $i$ is played, the
feedback about the reward outcome is available after $\delta_i$ time
steps. Once again the budgets of the arms is encoded in the respective
state spaces.  In this section we assume that there are no switching
or traversal dependent costs.  Given the algorithms and analysis we
propose for handling delayed feedback, such additional constraints can be handled using the ideas of the
preceding sections.

From an analysis standpoint, the idea of truncation is not
(immediately) useful. This is because a policy can (and should) be
``back-loaded'' in the sense that if we consider a single arm, most of
the plays are made towards the end of the horizon (possibly because
the policy is confident that the reward of this arm is large and well
separated from the alternatives). Therefore truncating the horizon (by
any factor) may cause these good plays to be eliminated and as a
consequence we would not have any guarantee on the expected reward.
In what follows we will introduce two techniques that avoid the
problem mentioned:

\begin{enumerate}[(a)]\parskip=0in
\item A {\em Delay Free Simulation}; where at some point of time in the policy, the delays become 
irrelevant.
\item A {\em Block Compaction Strategy}; where the plays of a policy are moved earlier in time.
\end{enumerate}

Of course, both of these ideas lose optimality, but interestingly that
loss of optimality can be bounded. These two techniques are similar,
they both increase the number of plays, and yet are useful in very
different regimes. The delay free simulation idea, by itself gives us
a $(2(1+\epsilon) + 32(y+y^2))$-approximation for any $\epsilon \in
(0,1]$ where $y=\max_i \delta_i/\sqrt{T}$. This implies that when the
delay is small, the result is very close to the finite horizon MAB
result discussed in Section~\ref{finite} except that we are not using
irrevocable policies. For $\max_i \delta_i \leq \sqrt{T}/50$ this 
approximation ratio is less than $3$.
However this idea ceases to be less useful by itself as $\max_i
\delta_i$ increases. 
The block compaction strategy (in conjunction with
the delay free simulation) allows an $O(1)$ approximation even when
$\delta_i$ are large as long as $21(2\delta_i+1) (1+\log \delta_i) \leq T$ for all $i$.

Note that it is not immediate how to construct good policies for
$\max_i \delta_i$ is a non-trivial constant (say $100$) or even the
larger regime without losing factors proportional to $\max_i
\delta_i$. Thus it is doubly interesting that these two regimes
naturally emerge from the analysis of approximation factors -- these
two regimes expose two different scheduling ideas. We come a full
circle from the discussion in Ehrenfeld~\cite{sched}, where the
explicit connections between stopping rules for delayed feedback in
the two bandit setting and scheduling policies were considered.

\paragraph{The Overall Plan and Roadmap:} We will bound the optimum
using single arm policies. In both regimes of $\max_i \delta_i$,
instead of solving the relaxations, we first show that there exists a
subclass of policies which (i) have much more structure, (ii) preserve
the approximation bound with respect to the globally optimum policies
and (iii) are easy to find. We first present the basic notation and
definitions in Section~\ref{basicdelay}. We then discuss the regime
where $\max_i \delta_i \leq \sqrt{T}/10$ in Section~\ref{smalldelta}
and prove the relevant approximation bound in
Lemma~\ref{smalldeltalemma}, under the assumption that we can find the
suitable policies.  In the interest of the presentation we do not
immediately show how to compute the policies but we discuss the regime
where $\delta_i$ are larger (but $21(2 \delta_i +1) (1+\log \delta_i)
\leq T$ for all $i$) next in Section~\ref{largedelta} and prove the
approximation bound in Lemma~\ref{bigdeltalemma}, under the same
assumption that we can find the suitable policies. Finally we show how to compute the
policies in polynomial time, using a linear program in
Section~\ref{proof:blocklemma}. That linear program can also be solved
to a $(1+\epsilon)$-approximation efficiently using ideas in the
previous sections.  As a consequence,

\begin{theorem}
\label{thmone}
We can approximate the finite horizon Bayesian MAB problem with
delayed feedback to a $2(1+\epsilon) + 32(y+y^2)$ approximation where
$y=\max_i \delta_i/\sqrt{T}$ for any $\epsilon>0$ in polynomial time.
If the delays are small constants or $o(\sqrt{T})$ then this result
implies a $(2+ o(1))$-approximation which seamlessly extends
Theorem~\ref{unitmab}.
For $\max_i \delta_i \leq T/(48 \log T)$ we provide a 
$O(1)$ approximation in polynomial time. 
\end{theorem}

\subsection{Weakly Coupled Relaxation and Block Structured Policies}
\label{basicdelay}

\paragraph{Single arm policies:} In the case of delayed
feedback, describing a single-arm policy is more complicated. This
policy is now a (randomized) mapping from the current state of the arm
to one of the following actions: (i) make a play; (ii) wait some
number of steps (less or equal to $T$), so that when the result of a
previous play is known, the policy changes state; (iii) wait a few
steps and make a play (without extra information); or (iv) quit.

\begin{definition} 
Let $\C_i(T,\delta_i)$ be the set of all single-arm
policies over a horizon of $T$ steps.
\end{definition}

As in Section~\ref{finite}, we will use the LP system 
(\lpdelay) to bound of the
reward of the best collection of single-arm policies - the goal will
be to find one policy per arm so that the total expected number of
plays is at most $T$ and the expected reward is maximized. This is
expressed by the following relaxation:

\[ \lpdelay =  \mbox{Max}_{\{P_i \in \C_i(T,\delta_i)\}} \left\{  \sum_i R(P_i) \ |\   \sum_i  \T(P_i)  \leq  T \right \} \]

 The
{\em state} of the system is now captured by not only the current
posterior $u \in \S_i$, but also the plays with outstanding feedback
and the remaining time horizon. Note that the state encodes plays with
outstanding feedback, and this has size $2^{\delta_i}$, which is
exponential in the input. Therefore, it is not even clear how to solve
(\lpdelay) in polynomial time, and we use (\lpdelay) as an initial 
formulation for the purposes of an upper bound.

\begin{definition}
  A single-arm policy is said to be {\em Block Structured} if the
  policy executes in phases of size $(2\delta_{i}+1)$. At the start of
  each phase (or block), the policy makes at most $\delta_{i}+1$
  consecutive plays. The policy then waits for the rest of the block
  in order to obtain feedback on these plays, and then moves to the
  next block. A block is defined to be {\em full} if exactly
  $\delta_i+1$ plays are made in it. Let the class of Block Structured policies for arm $i$ over a horizon $T$ be $\C^b_i(T,\delta_i)$.
\end{definition}

\noindent 
We first show that all single-arm policies can be replaced with block
structured policies while violating the time horizon by a constant
factor. The idea behind this proof is simple -- we simply insert
delays of length $\delta_i$ after every chunk of plays of length
$\delta_i$.

\begin{lemma}
\label{lemone}
Any policy $\P \in \C_i(T,\delta_i)$ can be converted it to a Block Structured policy  $\P' \in \C^b_i(2T,\delta_i)$ such that $R(\P) \leq R(\P')$  and $\T(\P') \leq \T(\P) $ (note that the horizon increases).
\end{lemma}
\begin{proof}
  We can assume that the policy makes a play at the very first time
  step, because we can eliminate any wait without any change of
  behavior of the policy.

  Consider the actions of $\P$ for the first $\delta_{i}+1$
  steps, the result of any play in these steps is not known before all
  these plays are made. An equivalent policy $\P'$ simulates $\P$
  for the first $\delta_{i}+1$ steps, and then waits for $\delta_{i}$
  steps, for a total of $2\delta_{i}+1$ steps.  This ensures that
  $\P'$ knows the outcome of the plays before the next block begins.

Now consider the steps from $\delta_{i}+2$ to $2\delta_{i}+3$ of
  $\P$. As $\P$ is executed, it makes some plays
  possibly in an adaptive fashion based on the outcome of the plays in
  the previous $\delta_{i} + 1$ steps, but not on the current
  $\delta_{i}+1$ steps.  $\P'$ however knows the outcome of the
  previous plays, and can simulate $\P$ for these $\delta_{i}+1$
  steps and then wait for $\delta_{i}$ steps again.  It is immediate
  to observe that $\P'$ can simulate $\P$ at the cost of increasing
  the horizon by a factor of $\frac{2\delta_{i}+1}{\delta_{i}+1} < 2$.
  Observe that in each block of $2\delta_{i}+1$, $\P'$ can
  make all the plays consecutively at the start of the block without
  any change in behavior. The budgets are also respected in this process. 
  This proves the lemma.
\end{proof}

We now show how to find a set of block structured policies $\{\P_i\}$
such that each has horizon at most $2T$ and together satisfy the
following properties $\sum_i R(\P_i) \geq \lpdelay$ and $\sum_i
\T(\P_i) \leq KT$. Observe that by Lemma~\ref{lemone}, such policies
exist.

\begin{lemma}
\label{blocklemma}
  We find a collection of policies $\{\P_i | \P_i \in
  \C^b_i(2T,\delta_i)\}$ such that $\sum_i R(\P_i) \geq
  \lpdelay/(1+\epsilon)$ and $\sum_i \T(\P_i) \leq KT$ for any
  $\epsilon \in (0,1]$ in time polynomial in $\sum_i \S_i,T$ (using a
  linear program or a fast approximate solution for such).
\end{lemma}

\noindent We relegate the proof of Lemma~\ref{blocklemma} to
Section~\ref{proof:blocklemma} and continue with the algorithmic
development.

\subsection{Delay Free Simulations and Small Delays}
\label{smalldelta}
In this section we discuss the case when the delays are small in comparison to the horizon, i.e., $\max_i \delta_i \leq \sqrt{T}/50$. We next introduce the delay 
free simulation idea.
Let $r=\sqrt{T}/(2 \max_i \delta_i)$ and $\gamma=4r(\max_i
\delta_i)^2/T$.  The choice of these parameters will become clear
shortly.

\begin{definition}
Define a policy $\P_i$ to be in a delay-free mode at time $t$ 
if at time $t$ the policy $\P_i$ makes a play, but 
uses the outcome of a play made at time $t-\delta_i$ (or earlier). The outcome of the current play is available at time $t+\delta_i$ and can 
be used subsequently. 
\end{definition}

\noindent Observe that the truncation idea of
Theorem~\ref{statetheorem} applies to a policy once it is in a delay
free mode.

\begin{definition}
Define a policy $\P_i$ to be a $(\gamma T,T)$-horizon policy if the policy is block structured on all decision paths till a horizon of $\gamma T$ and subsequently makes at most $T$ plays in a delay free mode.
A pictorial depiction of a $(\gamma T,T)$-horizon policy is shown if Figure~\ref{twohorizonfig}.
\end{definition}

\begin{figure*}[htbp]
\begin{center}
\subfigure[]{
\includegraphics[scale=0.5]{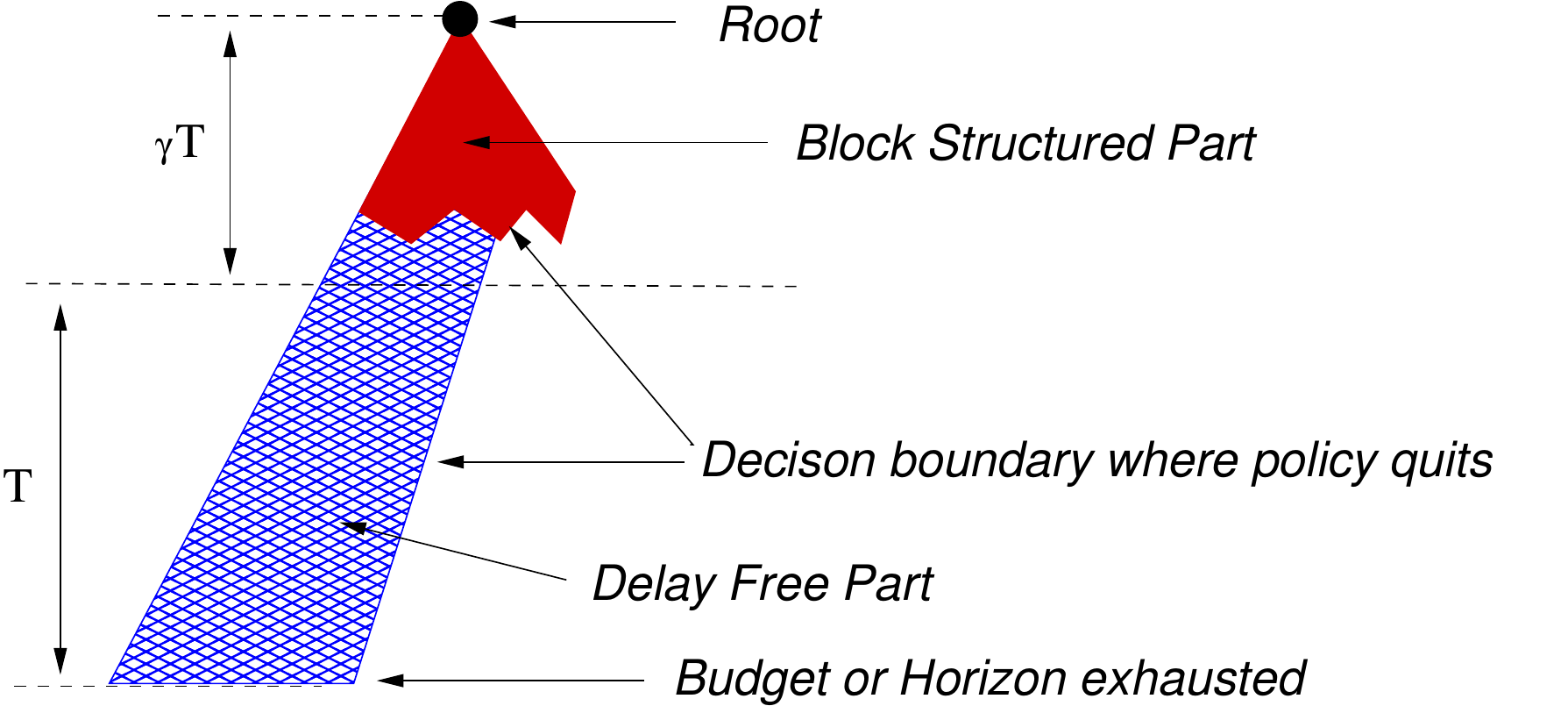}
\label{twohorizonfig-a}
}
\subfigure[]{
\includegraphics[scale=0.5]{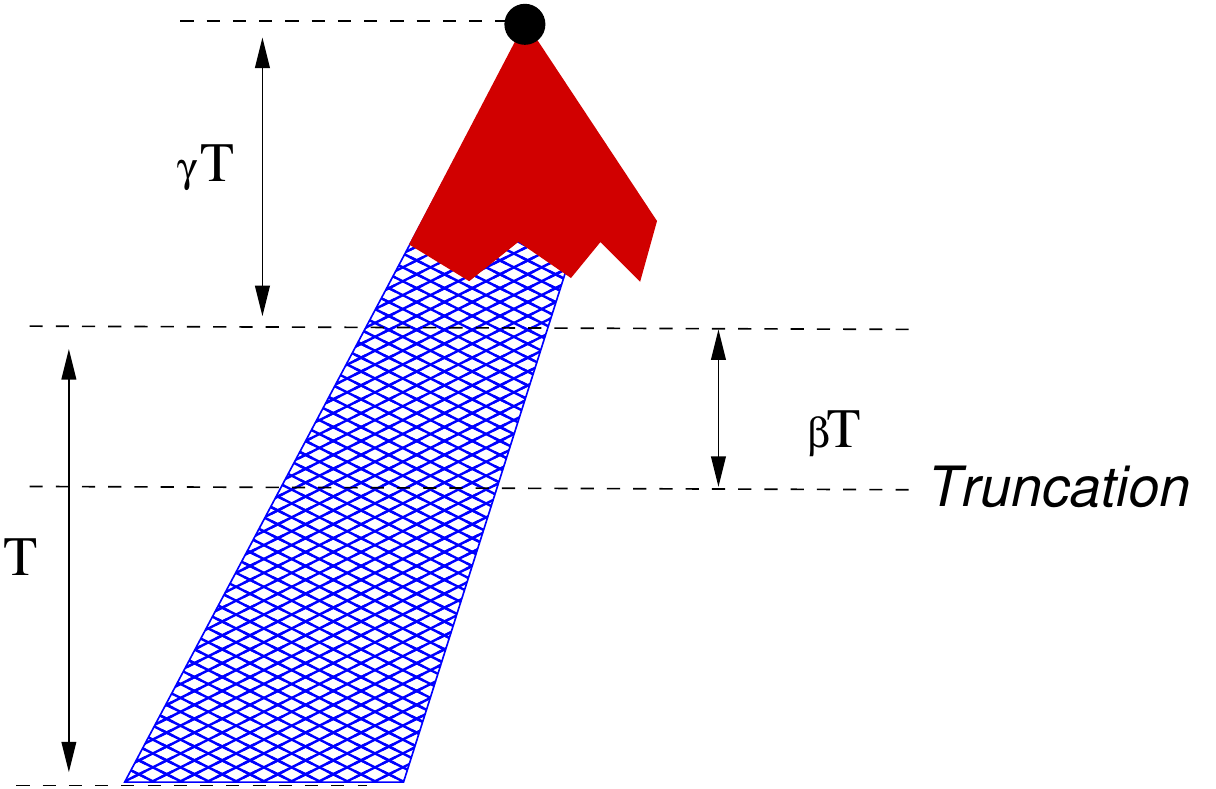}
\label{twohorizonfig-b}
}
\caption{\label{twohorizonfig} Part (a) shows an example $(\gamma
  T,T)$-horizon policy; the part marked in solid red is block
  structured and the patterned blue part is in delay free mode.  Part
  (b) shows the same policy when it is truncated at horizon $(\gamma
  +\beta)T$ where $\beta \in (0,1)$. 
  In the accounting step we 
  ensure that to receive any reward from the policy, the entire block
  structured part has to be played, i.e., $\beta>0$. Subsequently 
  we can argue that the reward of the policy in (b) is at least $\beta$ times that of the policy in (a) -- note that the argument is made on a path-by-path basis and we do not need a strong characterization of the decision boundary.}
\end{center}
\end{figure*}

\begin{lemma}
\label{lemtwo} 
Given any block structured policy $\P \in \C^b_i(2T,\delta_i)$ we can  construct a $(\gamma T,T)$-horizon policy $\P'$ such that  $R(\P) \leq R(\P')$ and $\T(\P') \leq (1+ \frac{1}{r}) \T(\P)$. 
\end{lemma}
\begin{proof}
  Consider the first time the policy $\P$ makes $r\delta_i$ plays on
  some decision path. Define a new policy $\P'$ as follows: It is
  identical to $\P$ till the end of the block that contains the $r
  \delta_i$-th play and then it makes $\delta_i$ additional plays. Let
  these plays be $z(1),z(2),\ldots,z(\delta_i)$.  Note that $\P$ was
  not making any plays and waiting for the end of the block. After
  this point, consider the $t$-th play made by $\P$ where
  $t\leq\delta_i$ -- the policy $\P'$ makes the play but simply uses
  the outcome of $z(t)$ which is known by now. The outcome of the
  $t$-th play is not known (or used) and would only be used for the
  $t+\delta_i$-th play. Note that since the outcome of two plays are
  from the same underlying distribution we can couple the outcomes to
  be identical as well. If $\P$ decides to stop execution, then $\P'$
  stops execution as well.
Clearly, $\P'$ makes $\delta_i$ additional plays than $\Pp$ on any
decision path. Since $\P$ made at least $r \delta_i$ plays by
then, this shows that $\T(\P') \leq (1+ \frac{1}{r}) \T(\P)$. Since
the execution of $\P'$ is coupled play-by-play with the execution of
$\P$, it is clear that $R(\P) \leq R(\P')$.

Observe that once $\P'$ has switched to this ``delay free mode'' then
$\P'$ can now simply ignore the blocks, but would make at most $T$
plays. This switch to a delay free mode must occur within the first $r
\delta_i$ blocks since each block must have one play -- and given that
the blocks account for $2\delta_i + 1$ units of time, the switch
occurs within an horizon of $r\delta_i(2\delta_i + 1) \leq
4r\delta^2_i \leq \gamma T$.
\end{proof}

\paragraph{Using the $(\gamma T,T)$-horizon policy:} 
Using Lemma~\ref{blocklemma} we can find a collection of block
structured policies $\{\P_i\}$ such that $\sum_i R(\P_i) \geq
\lpdelay/(1+\epsilon)$ and $\sum_i \T(\P_i) \leq KT$ for any
$\epsilon>0$. Using Lemma~\ref{lemtwo} we now have a collection of
$(\gamma T,T)$-horizon single arm policies $\{\P'_i\}$ such that
$\sum_i R(\P'_i) \geq \lpdelay/(1+\epsilon)$ and $\sum_i \T(\P'_i)
\leq (1+\frac1r) KT$. But before such, consider the following
scheduling policy given in Figure~\ref{fh3}.

\begin{figure*}[ht]
\fbox{
\begin{minipage}{6.0in}
{\small
{\bf The Scheduling Policy for Delayed Feedback when $\max_i \delta_i/\sqrt{T}$ is small}
\begin{enumerate}
\item Obtain a collection of $(\gamma T,T)$-horizon (as described by Lemma~\ref{lemtwo}) single-arm policies $\P'_i$ such that 
$\sum_i R(\P'_i) \geq \lpdelay/(1+\epsilon)$ for some $\epsilon \in (0,1]$ and $\sum_i \T(\P'_i) \leq (1+\frac1r) KT$ where $r>0$. Note $\gamma \leq 1/4$.
\item Order the arms in order of $\frac{R(\P'_i)}{\T(\P'_i)}$, a lower order indicates a higher priority. \label{needno5}
\item Let $\alpha=\frac{1-\gamma}{1+\frac1r}$.
For each $i$, create policy $\P''_i$ which chooses to plays $\P'_i$ 
with probability $\alpha$ and with probability $(1-\alpha)$ is the null policy. \label{alphastep}
\item A policy can be {\bf active} or {\bf passive}. 
Initially all policies are active.
\item A policy can be {\bf ready} or {\bf waiting}. Initially all non-null 
policies are ready. Note that some $\P''_i$ will not be ready 
since they are the null policy.
\item Inspect the {\bf active} 
policies in the order specified in Step~\ref{needno5}
and find the first $K$ policies (the highest priority ones) which are {\bf ready}, and make the 
corresponding plays.
\item Based on feedback received for different arms (these could be 
arms that were played long ago) update the state of all the arms.
\begin{enumerate}
\item If the policy for an arm quits then that policy is marked {\bf passive}.
\item If a policy has no plays for which it is waiting for feedback 
(and has not decided to quit) the policy is marked {\bf ready}, otherwise the policy remains {\bf waiting}.
\end{enumerate}
\item If the horizon $T$ is reached, the overall policy stops execution. \label{horizonstep}
\end{enumerate} }
\end{minipage}
}
\caption{A Scheduling Policy for Finite Horizon MAB with delayed feedback. Note that this policy combines elements of the policies in Figure~\ref{fh} and Figure~\ref{fh2}.\label{fh3}}
\end{figure*}

\begin{lemma}
\label{smalldeltalemma}
The expected reward $R$ of the scheduling policy in Figure~\ref{fh3} is at least $\lpdelay/(2(1+\epsilon) + 32(y+y^2))$ where $y=\max_i \delta_i/\sqrt{T}$. For $\max_i \delta_i \leq \sqrt{T}/50$ the approximation ratio does not exceed $3$ for suitably small $\epsilon \in (0,0.1]$. 
\end{lemma}
\begin{proof}
Again let $T_j$ be the actual number of plays of a policy for arm $j$. 
Observe that the policy for arm $i$ does not affect the policy for arm $j$ if $j<i$ in the order defined in Step~\ref{needno5} in the policy shown in
Figure~\ref{fh3}. 
Now consider the 
policy $\P''_i$.
If a policy does not
start at or before time $(1-\gamma)T$ we disregard its entire reward (see the pictorial representation of the accounting in Figure~\ref{twohorizonfig}) because the policy will not finish executing the entire block structured part.
Now consider the 
policy $\P''_i$. Observe that 
the delay free part of policy $\P''$ is truncated by a factor {\bf at most} 
$ \frac1T \left( T - \frac1{KT} \sum_{j<i} T_j - \gamma T \right)$.
Note that the policy $\P''_i$ may start before all the plays in
$\sum_{j<i} T_j$ are made, because some of those policies were waiting
for feedback. Therefore the contribution from $\P''_i$ conditioned on $T_1,T_2,\ldots,T_{i-1}$ is at least in expectation

\[ \frac1T \left( T - \frac1{K} \sum_{j<i} T_j - \gamma T \right)
R(\P''_i) \] 
Again note that the actual reward will be higher because
the contribution cannot be negative -- whereas we are ignoring the
case when $T \leq \frac1{K} \sum_{j<i} T_j + \gamma T$. Moreover 
the contribution from the first $\gamma T$ steps is always present.
However the above 
accounting suffices for our bound. The expected reward (summed over
all $i$ and the respective conditionings removed exactly as in the
proof of Lemma~\ref{greedyorder}) we get:
\[ R \geq \frac1T \sum_i \left( T - \frac1{K} \sum_{j<i} \T(\P''_j) - \gamma T \right) R(\P''_i) \geq (1-\gamma) \sum_i \left( 1 - \frac{\sum_{j<i} \T(\P''_j)}{KT(1-\gamma)} \right) R(\P''_i) \]
But since $\sum_i \T(\P''_i) \leq KT(1-\gamma)$, by exact same argument 
as in Lemma~\ref{greedyorder} we have:
\begin{equation}
R \geq (1-\gamma) \frac12 \sum_i R(\P''_i) \geq (1-\gamma) \frac12 \alpha \sum_i R(\P'_i) \geq \frac{(1-\gamma)\alpha}{2(1+\epsilon)} \lpdelay \label{delequation}
\end{equation}
Using $\alpha=\frac{1-\gamma}{1+1/r},\gamma\leq \frac14$ and $\epsilon\leq 1$the approximation ratio of $\lpdelay/R$ is at most 
\[ \frac{2(1+\epsilon)(1+\frac1r)}{(1-\gamma)^2} \leq 2(1+\epsilon)(1+1/r)(1+2\gamma)
\leq 2(1+\epsilon) + \frac8r + 8\gamma + \frac{8\gamma}{r}  
\]
Since $\gamma=4r(\max_i \delta_i)^2/T$ the above ratio is 
$2(1+\epsilon) + \frac{32(\max_i \delta_i)^2}{T} + 8\left(\frac1r + \frac{4r (\max_i \delta_i)^2}{T} \right)$
which is minimized at $r=\sqrt{T}/(2\max_i \delta_i)$ giving an 
approximation ratio
of $2(1+\epsilon) + 32(y+y^2)$ where $y=\max_i \delta_i/\sqrt{T}$
which proves the Lemma.
\end{proof}

\noindent Lemma~\ref{smalldeltalemma} also
shows that for large $\max_i \delta_i$ the ratio grows as
$O\left(\frac{(\max_i\delta_i)^2}{T}\right)$, which indicates why
other ideas are needed. We now discuss the idea of block compaction.

\subsection{Block Compaction and Larger Delays}
\label{largedelta}
In this section we assume that $T\geq 21(2\delta_i+1)(1+\log \delta_i)$ for all $i$. Let $\gamma=\frac13$, $\rho=13$.
The next lemma states that we can also shrink the horizon by
increasing the number of plays --- which is intuitively equivalent to
moving the plays forward in the policy.

\begin{lemma}
\label{lemthree}
For any $\rho > 1$ given a policy $\P_i \in \C^b_i(2T,\delta_i)$, then there exists a $(\gamma T,T)$-horizon policy (as defined in Lemma~\ref{lemtwo}) $\P'_i$ such that $R(\P_i) \leq R(\P'_i)$ and $\T(\P'_i) \leq \rho \T(\P_i)$.
\end{lemma}
\begin{proof}
  Consider the execution of $\P_i$. Without loss of generality we can
  assume that all blocks in this policy have at most
  $\delta_i/\rho$ plays, otherwise the policy can transition to
  the delay free mode increasing the plays by a factor of $\rho$.
  The total number of plays in the delay free mode can be at most $T$.
  In the remainder of the proof we will consider the prefix of the
  policy where no block has more than $\delta_i/\rho$ plays.
  We consider the execution of the original policy, and show a coupled
  execution in the new policy $\P'_i$ so that if on a particular
  decision path, $\P_i$ used $(\#b)$ blocks, then the number of blocks
  on the same decision path in $\P'$ is $\lfloor (\#b)/7 \rfloor +
  1+\log \delta_i$.  Observe that since $(\#b)$ is at most
  $2T/(2\delta_i + 1)$; the time horizon required for $\lfloor
  (\#b)/7 \rfloor + 1+\log \delta_i$ blocks is at most

  \[ \left( \frac{(\#b)}{7} + 1 + \log \delta_i \right) (2\delta_i + 1)
  < \frac{2T}{7} + 3 \delta_i (1+\log \delta_i) < \gamma T -
  \left(\frac{T}{21} - 3 \delta_i (1+\log \delta_i) \right) \leq \gamma T
  \] 
  for the setting of $\gamma=1/3$ and the assumption on
  $T \geq 21(2 \delta_i + 1) (1+\log \delta_i)$ for all $i$. Therefore proving
  that the number of blocks is $\lfloor (\#b)/7 \rfloor + 1+\log
  \delta_i$ is sufficient to prove the lemma.
  
  We group blocks of the policy $\P_i$ into size classes; size class
  $s$ has blocks whose number of plays lies in $[2^s,2^{s+1})$.  We
  couple the executions as follows: Consider the execution of $\P_i$
  and define the new policy $\t{\P}_i$ as follows. 

  \smallskip Suppose there were $x$ plays in $\P_i$ in the first
  block, and the size class of this block is $s$.  Then $\t{\P}_i$
  makes $\rho x=13x$ plays in this block. The policy $\t{\P}_i$
  uses the the outcomes of the extra $12x$ plays made in the first 
  block to simulate the behavior of $\P'_i$ as below:
  \begin{enumerate}[(a)]
  \item At any point of time $\t{\P}_i$ has an excess $z$
    number of plays which it has made, and knows the outcome but has
    not utilized the outcome yet. It also has a multi-set ${\mathbb S}$ 
    of types which it can simulate. Initially $z=12x$ and 
    ${\mathbb S} = \{s,s,s,s,s,s\}$. Observe that we will maintain 
   the invariant that $z$ is larger than the number of plays required to 
   play the all types of blocks in ${\mathbb S}$. 
  \item Suppose the next block according to the decisions made by 
  $\P_i$ (and simulated by $\t{P}_i$) makes $x'$ plays corresponding to 
  type $s'$.
\begin{enumerate}[(i)]
\item If $s' \in {\mathbb S}$ then we use the {\em first} $x'$
  outcomes we already know and set $z\leftarrow z -x'$ and ${\mathbb
    \S} \leftarrow {\mathbb S} - \{s'\}$ (removing just one copy of
  $s'$ from the multiset). Note that we do not make any plays and do
  not suffer the $2\delta_i + 1$ time steps for this block.
\item If $s' \not \in {\mathbb S}$ then we make $\rho x'=13x'$ plays. 
  We now again use the first $x'$ of the outcomes -- if $z\leq x'$ then 
  we use all the old $z$ outcomes and use $x'-z$ new outcomes. Otherwise if $x'<z$ then we use the $x'$ old outcomes. 
In either case the total number of excess outcomes we have stored is $z+12x'$. We set $z\leftarrow z + 12x'$ and ${\mathbb
    \S} \leftarrow {\mathbb S} \cup \{s',s',s',s',s',s'\}$.
\end{enumerate}
\end{enumerate}

Observe that we are maintaining the invariant that the outcomes of the
plays are used only after they are known and in first-in first-out
fashion. Moreover making a play corresponding to $x$ plays of type $s$
will definitely generate $12x$ extra plays which can be used to
simulate the next $6$ blocks of type $s$ in the decision path {\em
  irrespective of when those blocks will occur}. Observe that since
the underlying reward distribution is fixed, these plays will be
stochastically identical in both policies $\P_i,\t{\P}_i$. Moreover it
it easy to see that for any decision path in $\P_i$ with $(\#b)$
blocks, the corresponding path in $\t{\P}_i$ has $\lfloor (\#b)/7
\rfloor + 1+\log \delta_i$ blocks -- where the division by $7$ arises
from the current block (of $x$ plays) and $6$ other blocks of type $s$
(at most $2x$ plays each). The additive term of $1+\log \delta_i$
corresponds to the total number of size types and arises from the fact
that we may not have any blocks of a particular type $s$ on the
decision path. Observe that this is a worst case bound and holds for
any decision paths. 

Therefore we can immediately say that in $\t{\P}_i$ either we will
quit the policy within a horizon of $\gamma T$ or would have started
to play in a delay free manner.

The proof requires one final observation -- $\t{\P}_i$ already uses
the plays in a first-in first-out fashion and respects the block
boundaries. It differs from a $(\gamma T,T)$-horizon policy in that it
does not use all the information it already has. However then there
exists a block structured policy that uses the information and
maximizes the reward while maintaining the same number of plays! This
block structured policy 
is the policy $\P'_i$.
\end{proof}

\noindent We now observe that there exists $(\gamma T, T)$-horizon
policies $\P'_i$ such that $\sum_i R(\P'_i) \geq \lpdelay$ and $\sum_i
\T(\P'_i) \leq \rho KT$. As we will see in section~\ref{proof:blocklemma}, for
any $\epsilon \in (0,1]$ we can find $(\gamma T, T)$-horizon policies
$\{\P'_i\}$ that satisfy $\sum_i R(\P'_i) \geq \lpdelay/(1+\epsilon)$
and $\sum_i \T(\P'_i) \leq \rho KT$. We can now apply the scheduling policy 
in Figure~\ref{fh3} (with changes to Steps~1 and~3) setting $\alpha=(1-\gamma)/\rho$. The next lemma is immediate:

\begin{lemma}
\label{bigdeltalemma}
Suppose $T\geq 21(2 \delta_i +1) (1+\log \delta_i)$ for all $\delta_i$.
Consider the policy which is identical to the policy in Figure~\ref{fh3} except that in Step~1 we find $(\gamma T, T)$-horizon policies $\{\P'_i\}$ that satisfy $\sum_i R(\P'_i) \geq \lpdelay/(1+\epsilon)$ and $\sum_i \T(\P'_i) \leq \rho KT$; and in Step~3 we set $\alpha=(1-\gamma)/\rho$. This new scheduling policy gives an expected reward $R'$ which is at least $\lpdelay/(119(1+\epsilon))$.
\end{lemma}
\begin{proof}
This proof is identical to the proof of Lemma~\ref{smalldeltalemma} up to Equation~(\ref{delequation}) in concluding 
\[ R' \geq \frac{(1-\gamma)\alpha}{2(1+\epsilon)}\lpdelay \]
The rest of the lemma follows from $\gamma=1/3$ and $\alpha=(1-\gamma)/\rho = 2/39$.
\end{proof}

\noindent The interesting aspect of Lemma~\ref{bigdeltalemma} is that
the approximation factor remains $O(1)$ even when $21(2\delta_i+1) (1+\log \delta_i) \approx T$ which is rather large delays.

\subsection{Constructing Block Structured Policies: Proof of Lemma~\ref{blocklemma}}
\label{proof:blocklemma}
\renewcommand{\SS}{{Super}} 
\newcommand{\tp}{{\tilde{p}}} 

We now show how to find a set of block structured policies $\{\P_i\}$
such that each has horizon at most $2T$ and together satisfy the
following properties $\sum_i R(\P_i) \geq \lpdelay$ and $\sum_i
\T(\P_i) \leq KT$. Observe that by Lemma~\ref{lemone}, such policies
exist.  Such a policy for arm $i$ operates in blocks of length $2
\delta_i + 1$, and over a horizon of $2T$ steps. In each block, it
makes at most $\delta_i/\rho$ plays, and then waits for the feedback
of these plays.
For any state $\sigma$ of arm $i$, the policy decides on the number $\ell$ of consecutive plays made in this state. Define the following quantities. These quantities are an easy computation from the description of $\S_i$, and the details are omitted.
\begin{enumerate}
\item Let $r_i(\sigma,\ell)$ denote the expected reward obtained when $\ell$ consecutive plays are made at state $\sigma$, and feedback is  obtained at the very end.
\item Let $p_i(\sigma,\sigma',\ell)$ denote the probability that if  the state at the beginning of a block is $\sigma$, and $\ell$ plays are  made in the block, the state at the beginning of the next block is  $\sigma'$. 
\end{enumerate}

We will formulate an LP to find a collection of randomized block-structured
policies $\Pp_i$. For each $i$, define the following variables over the decision tree of
the policy $\Pp_i$ and the system $(LPD2)$.
\begin{itemize}\parskip=0in
\item $x_{i\sigma}$: the probability that the state for arm $i$ at the start of a block is $\sigma$.
\item $y_{i\sigma \ell}$: probability that $\Pp_i$ makes $0\leq \ell \leq \delta_i$ consecutive plays starting at state  $\sigma$.
\end{itemize}
\begin{eqnarray*} 
& & LPD2 =  \mbox{Max}  \sum_{i,\sigma,\ell} r_i(\sigma,\ell) y_{i\sigma \ell}  \\
& & \begin{array}{rcll}
\sum_{\sigma} x_{i\sigma} & \le & 1 & \forall i \\
\sum_{\ell} y_{i \sigma \ell} & \le & x_{i \sigma} & \forall i, \sigma \\
\sum_{\sigma,\ell} y_{i \sigma \ell} p_i(\sigma, \sigma', \ell) & = & x_{i \sigma'} & \forall i, \sigma' \\
\sum_{i,\sigma,\ell} \ell y_{i \sigma \ell} & \le &  KT & \\
x_{i \sigma}, y_{i \sigma \ell} & \ge & 0  & \forall i, \sigma, \ell
\end{array} 
\end{eqnarray*}

The number of variables is at most $\sum_i |\S_i|\delta_i$ which is 
polynomial in $T$ and $\sum_i |\S_i|$. We have the following LP relaxation, 
which simply encodes finding one
randomized well-structured policy per arm so that the expected number
of plays made is at most $T$. The system (LPD2) returns a
collection of single arm policies; the policies are interpreted in
Figure~\ref{fig1a}. The interpretation is almost the same as in
Figure~\ref{fig1} except that there $\ell=1$.

\begin{figure*}[htbp]
\centerline{\framebox{\small
\begin{minipage}{6in}
\begin{itemize}
\item If the state at the beginning of a block is $\sigma$: 
\begin{enumerate}
\item Choose $\ell$ with probability $\frac{y_{i \sigma \ell}}{x_{i \sigma}}$, and make $\ell$ plays in the current block.
\item Wait till the end of the block; obtain feedback for the $\ell$ plays; and update state.
\end{enumerate}
\end{itemize}
\end{minipage}}
}
\caption{\label{fig1a} Single-arm policy returned by the solution of (LPD2).}
\end{figure*}

\noindent Observe that as a consequence of Lemma~\ref{lemone}, $LPD2 \geq \lpdelay$. Lemma~\ref{blocklemma} follows.

\section{\maxmab: Bandits with Non-linear Objectives}
\label{sec:maxmab}
In \maxmab\, at each step, the decision policy can play at most
$K$ arms but the reward obtained is the maximum of the values that are
observed. In other words, the policy plays at most $K$ arms each step,
but is only allowed to choose the arm with the maximum observed value
and obtain its reward.  Note that the states of all of the arms which
are played evolve to their respective posteriors, since all these arms
were observed.  The reward is clearly nonlinear (nonadditive) in the
plays -- and several issues arise immediately.

\paragraph{(a) Handling Budgets:} As discussed in
Section~\ref{budgets} we have a modeling choice in terms of how
budgets are handled. It is easy to conceive of two separate
application scenarios. In the first application scenario, every arm
that is played expends a budget that it its observed value -- this is
the \allpay\ model as discussed in Section~\ref{budgets}.
This is natural in scenarios where the observed value correlates with
the effort spent by the alternative. This case can easily be handled
by truncating the space $\S_i$ so that states $u \in \S_i$
corresponding to total observations larger than $B_i$ are discarded --
and this is the model we have discussed in all the previous sections
since we considered linear (additive) rewards (or $K=1$ for which
avoids this issue).

\smallskip In the second application scenario only the arm with the
maximum observed value is charged its budget -- which is the
\onlythemax\ model as discussed in Section~\ref{budgets}. This
accounting is natural in many strategic scenarios. Note that in this
budget model, if we have a good arm then we can easily resolve the
priors of $K-1$ other arms at a time while still getting the reward of
the good arm. When the good arm is exhausted, then we have much more
information about the other arms. Therefore the policies corresponding
to these arms will only exhaust the budget deeper into the horizon.
This runs counter to the intuition of the truncation! Yet, we can show
that up to a $O(1)$ factor the difference in accounting does not
matter.

\paragraph{(b) Simultaneous or One-at-a-time feedback:} Given that we
are playing at most $K$ arms then again it is not difficult to
conceive of two separate application scenarios. In the first
application scenario, at each time slot we decide on the set of arms
to play and then make the plays simultaneously -- we then receive the
respective feedbacks and finally choose the reward for the slot. This
is the \simult\ model.  In the second scenario we make the plays
one-at-a-time and get immediate feedback. We refer to this model as
\oneatatime\ feedback model.
Therefore
  in this scenario we can choose a reward even before we have played
  all the potential set of arms and move to the next time slot (only
  the arms which have been played gets updated to their corresponding
  posteriors). Clearly the second model allows for more powerful and more 
  adaptive policies in the \allpay\ budget model.
 
  \smallskip
  It may appear that in the \onlythemax\ model, this
  feedback issue is not relevant -- because once the set of arms are
  fixed in a single time slot there is no need to stop (choose reward)
  and move to the next time slot before observing all the values.
  However in the one-at-a-time feedback model the {\em set} of arms
  that can potentially be played in a time slot changes as the
  outcomes of the plays in that time slot are obtained. Therefore the
  one-at-a-time feedback model is more powerful even in the \onlythemax\ budget 
  model. The feedback and budget issues are
  two different and orthogonal modeling choices.

\paragraph{Roadmap:} We prove the following results. 
\begin{enumerate}[(i)]\parskip=0in
\item First we consider the observed value budget model.  We provide a
  factor $(4+\epsilon)$-approximation for the one-at-a-time feedback
  model in Section~\ref{simplest}.  We then show that the algorithm can
  be modified to work in the \simult\ model while providing an
  $O(1)$-factor reward of the \oneatatime\ feedback model in Section~\ref{harder}.
\item We consider the \onlythemax\ budget model in
  Section~\ref{hardest} and show that there exists a policy uses the
  \allpay\ accounting for budgets over a subset of the states in
  $\S_i$ and is yet within $O(1)$-factor of the best \oneatatime\
  feedback policy in the \onlythemax\ model over $\S_i$. The results
  extend to the \simult\ model as in Section~\ref{harder} with an
  appropriate loss of approximation.
\end{enumerate}

\noindent In terms of techniques, we introduce the basic ideas and
single arm policies in Section~\ref{simplest}. We then introduce a
``throttling'' notion in the scheduling of the different policies in
Section~\ref{harder}. We also show how this notion extends to handle
concave functions other than the maximum, for example knapsack type
constraints. Finally in Section~\ref{hardest} we show that there
exists a subset of states $\S'_i \subseteq \S_i$ for each arm such
that if we restrict ourselves to $\S'_i$ then the \simult\ 
and \allpay\ accounting suffices for a $O(1)$-approximation to
the optimum \oneatatime\ feedback policy with \onlythemax\ accounting.

\subsection{\oneatatime\ Feedback and \allpay\ Models}
\label{simplest}
In this section we assume that the
budget model is the \allpay\ model, that it, the total
value of all observations that can be derived from an arm is bounded.
This can be easily incorporated in the statespace $\S_i$ as in the
previous sections. We now consider the \oneatatime\ feedback model.

\paragraph{Single Arm Policies:} 
In this case, the single-arm policy for arm $i$ takes one of several
possible actions (possibly in a randomized fashion) for each possible
state of arm $i$ at any time step: (i) Stop execution; (ii) Play the
arm but do not choose it; (iii) Play the arm and choose it if the
observed value from a play is $q$, and obtain reward
$q$; and (iv) Do not play the arm, {\em i.e.}, wait for the next time step.

\begin{definition}
  For any single-arm policy $\P_i$, let $N(q,\P_i)$ denote the
  expected number of times the policy chooses arm $i$ (hence obtaining
  its reward) {\bf and} the reward is $q$. Let $\h{R}(\P_i) = \sum_q q
  N(q,\P_i)$; the total expected reward of the policy.  Recall $\T(\P_i)$
  is the expected number of plays of the policy.
\end{definition}

The goal of the following relaxation will be to find a (randomized)
single-arm policy $\P_i$ for each arm so that two constraints are
respected: The expected number of times any arm is chosen is at most
$T$ (since we allow only one arm to be chosen per time step); and the
expected number of plays made is at most $KT$ (since there are at most
$K$ plays per step).  Consider the following LP:
{\small
\[ \lpmaxmab = \mbox{Max}_{\{\P_i\}} \left\{\sum_i  \sum_q q N(q,\P_i)\  | \ \sum_i \sum_q   N(q,\P_i) \leq T \quad \mbox{\bf and} \quad \sum_i \T(\P_i)\le  KT \right\}\]
}
\begin{lemma}
\label{lbound1}
\lpmaxmab\ is an upper bound on $OPT$, the value of the optimal policy for the 
one-at-a-time feedback model.
\end{lemma}
\begin{proof}
  Consider the (randomized) single-arm policies obtained by viewing
  the execution of the optimal policy restricted to each arm $i$.  It
  is easy to check via linearity of expectation that these policies
  are feasible for the constraint in \lpmaxmab, since the expected
  number of choices made is at most $T$, and the expected number of
  plays made is at most $KT$. Further, the objective value of these
  policies is precisely $OPT$.
\end{proof}

\noindent We first consider the relaxation to a single constraint
$\sum_i \left(\sum_q N(q,\P_i) + \T(\P_i)/K \right) \leq 2T$ and then
taking the Lagrangian. The result is:
{\small
\begin{align*}
\modmaxmab(\lambda) & = \mbox{Max}_{\mbox{Policies } \{\P_i\}} \left\{ 2T \lambda +  \sum_i  \left(\sum_q (q - \lambda) N(q,\P_i) - \frac{\lambda}{K} \T(\P_i)\right) \right\} \\
& = 2T \lambda + \sum_i  \mbox{Max}_{\mbox{Policies } \{\P_i\}} \left(\sum_q (q - \lambda) N(q,\P_i) - \frac{\lambda}{K} \T(\P_i)\right) = 2T \lambda + \h{Q}_i
\end{align*}
}
\begin{lemma}
\label{lem:struct}
There exists a collection of optimum single arm policies $\{\P_i\}$ for $\modmaxmab(\lambda)$ 
which (a) do not wait (use the operation (iv) in the definition of single arm policies); (b)
always  choose the arm when the observed value is $q > \lambda $ and obtain
$q-\lambda$ reward; and (c) can be computed by dynamic programming in time $O(\sum_i \edge(\S_i))$.
\end{lemma}
\begin{proof}
$\modmaxmab(\lambda)$ has no constraints connecting the arms and has a 
separable objective, thus the presence or absence of other arms is immaterial to arm a $i$.
Therefore condition (a) follows -- there is no advantage to waiting.

Given a policy $\P_i$, if we alter the policy to always choose $q$ if
$q>\lambda$ then the objective can only increase, and therefore (b) follows.

Part (c) follows from the same style of bottom up dynamic programming
as in the proof of Lemma~\ref{mainlem}.  Recall that $X_{iu}$
corresponds to the posterior distribution of reward of arm $i$ if it
is observed in state $u \in \S_i$. Then if the gain of a node $u$ is defined by
{\small
\[
gain(u) = \max \left \{ 0, \left( \sum_{q > \lambda} \Pr[X_{iu}=q] (q- \lambda) \right) - \frac{\lambda}{K} + \sum_v \p_{uv} gain(v) \right \} \]
}
where the $\sum_v \p_{uv} gain(v)$ term is $0$ for all leaf nodes (since they have no children).
Note $\h{Q}_i =gain(\rho_i)$ where $\rho_i$ is the root of $\S_i$.
\end{proof}

\noindent The next theorem follows exactly using the same arguments as
in the proof of Theorem~\ref{blahtheorem}.

\begin{theorem} 
\label{blahtheoremmax} 
In time $O((\sum_i \edge(\S_i))\log (nT/\epsilon))$, we can compute 
randomized policies $\{\P_i\}$ such that $\h{R}(\P_i) \geq \lpmaxmab/(1+\epsilon)$ and 
$\sum_i \left( \sum_q N(q,\P_i) + \T(\P_i)/K \right) \leq 2T$.
\end{theorem}
\newcommand{\pr}{\mbox{Pr}}
\newcommand{\tail}{\mbox{\sc Tail}}

\noindent The next lemma is the key insight of the algorithm.
\begin{lemma}
\label{keymaxmab}
If the reward of a policy $\P_i$ is defined as $\h{R}(\P_i)= \sum_{q > \lambda} q N(q,\P_i)$ or $\tilde{R}(\P_i)= \sum_{q > \lambda} (q-\lambda) N(q,\P_i)$
for any $\lambda$, then the Truncation theorem (Theorem~\ref{statetheorem}) applies.
\end{lemma}
\begin{proof}
  The proof of Theorem~\ref{statetheorem} used a path by path
  argument. That argument is valid in this case as well.
Recall that $X_{iu}$ corresponds to the posterior distribution of
reward of arm $i$ if it is observed in state $u \in \S_i$. By Bayes'
rule, it is easy to check that $\pr[X_{iu}=q]$ itself is a Martingale
over the state space $\S_i$.

\begin{definition}
Given a random variable $X$ define  $\tail(X,\lambda) = \sum_{q \geq \lambda} \pr[ X=q] \cdot q$.
Define $\excess(X,\lambda) = \sum_{q > \lambda} \pr[ X=q] \cdot (q - \lambda)$.  
\end{definition}

\noindent It is easy to check that both $\tail(X_{iu},\lambda),\excess(X_{iu},\lambda)$ are martingales over $\S_i$.

\smallskip
\noindent For the case $\h{R}(\P_i)$ consider the execution of $\P_i$
conditioned on some choice of the underlying distribution $D_i$ drawn
from $\D_i$. Suppose the policy plays for $t$ steps on some decision
path. At each step, the reward obtained is drawn i.i.d. from $D_i$,
and so the expected reward is $\tail(D_i,\lambda)$ (which is drawn
from a non-negative distribution), and this quantity is independent of
rewards obtained in previous steps. Therefore, the expected reward is
$t \cdot \tail(D_i,\lambda)$. If the policy is restricted to execute
for $\beta T$ steps, this reduces the length of this decision path by at
most factor $1/\beta$, so that the expected reward obtained is at least $\beta t
\cdot \tail(D_i,\lambda)$. Taking expectation over all decision
paths and the choice of $D_i$ proves the claim.

The case of $\tilde{R}(\P_i)$ follows from the exact same argument applied to 
$\excess(D_i,\lambda)$.
\end{proof}

\noindent {\bf Remark:} Observe that different policies in $\{\P_i\}$ can have
different $\lambda$ since $\P_i$ is a randomized policy based on
$\lambda^+$ or $\lambda^-$. However the important aspect of
Lemma~\ref{keymaxmab} is that it holds for any $\lambda$. 

\paragraph{The Final Scheduling:} The algorithm is presented in Figure~\ref{fhmaxmabone}.

\begin{figure*}[ht]
\fbox{
\begin{minipage}{6.0in}
{\small
{\bf The Combined Final Policy}
\begin{enumerate}
\item For each $i$ independently choose $\P'_i$ to be $\P_i$ or the null policy with probability $\frac12$. \label{steptobereffed}
\item Order the arms in order of $\frac{\h{R}(\P'_i)}{\sum_q N(q,\P'_i) + \T(\P'_i)/K}$. This order would have been the same if we had used $\P_i$.
\label{neednomaxone}
\item Start with the first $K$ policies in the order specified in Step~\ref{needno}. 
These policies are active; the remaining policies are inactive.
\item In each time slot:
\begin{enumerate}
\item Consider the first active policy in the order specified and play it and update the 
posterior.
\item If the decision of that policy was to ``choose reward'' based on the observed outcome.
\item Update the state of this arm. If the decision in the new state is to quit then declare the next policy in the order specified in Step~\ref{neednomaxone} as active. We maintain $K$ active policies (unless there are less than $K$ policies which have not quit).
\item Repeat the above steps until either we have chosen a reward or have made $K$ different plays in this time slot.  
\item If the horizon $T$ is reached, the overall policy stops execution. 
\end{enumerate}
\end{enumerate} }
\end{minipage}
}
\caption{The Final Policy for \maxmab\ in one-at-a-time feedback model\label{fhmaxmabone}}
\end{figure*}

\begin{theorem}
\label{simplest:thm}
The expected reward of the final policy is at least $\lpmaxmab/(4(1+\epsilon))$. 
\end{theorem}
\begin{proof}
The proof is near identical to the proof of Theorem~\ref{unitmab}. Observe that 
{\small \[ \sum_i \left( \sum_q N(q,\P'_i) + \T(\P'_i)/K \right) = T \] }
and the horizon of the policy $\P'_i$ can be truncated to $T - Y_i$ where $Y_i$ is 
the number of slots in which $\P'_i$ is not allowed to play. Again if $T_j$ is 
the number of steps of policy $\P'_j$ and $n_j$ is the number of steps where policy $\P'_j$ chose a reward then
{\small
\[ Y_i \leq \sum_{j<i} \left( n_j + \frac{T_j}{K} \right) \qquad \Longrightarrow \qquad \E[Y_i] \leq \sum_{j<i} \left( \sum_q N(q,\P'_j) + \T(\P'_j)/K \right) \]
}
\noindent which immediately implies that using Lemma~\ref{keymaxmab} the expected reward is at least 
{\small
\[
\sum_i \E \left[ 1 - \frac{Y_i}{T} \right] \h{R}(\P'_i) \geq \frac12 \sum_i \h{R}(\P'_i)  \geq \frac14 \sum_i \h{R}(\P_i)\]
}
using the argument identical to the proof of Lemma~\ref{greedyorder}. The theorem follows.
\end{proof}

\subsection{Simultaneous Plays and \simult\ Model}
\label{harder}
\newcommand{\stall}{\mbox{\sc ThrottledCombine}}
\newcommand{\lagfirst}{\mbox{\sc LagFirstMax}}

We now show that we can make simultaneous plays; receive simultaneous feedback on all the arms that are played in a single time step; and still obtain a
reward which is an $O(1)$ factor of $\lpmaxmab$. We use the following
lemma;

\begin{lemma}
\label{truncate-within}
For $\beta,\alpha \in (0,1]$ in time $O((\sum_i \edge(\S_i))\log^2 (nT/\epsilon'))$ 
we can find a collection of single arm policies 
$\{\L_i\}$ such that (a) Each $\L_i$ has a horizon of at most $\beta T$;
(b) $\sum_i \T(\L_i) \leq \alpha KT$;
(c) $\sum_i \sum_{q} N(q,\L_i) \leq \alpha\beta T$; and
(d) $\sum_i \sum_q q  N(q,\L_i) \geq \alpha\beta\lpmaxmab/(1+\epsilon)$.
\end{lemma}
\begin{proof}
Consider the Lagrangian relaxation of just the first constraint with multiplier 
$\lambda(1)$;
{\small
  \[ \lagfirst = \mbox{Max}_{\{\P_i\}} \left\{\lambda(1) T + \sum_i
    \sum_q (q - \lambda(1)) N(q,\P_i)\ | \sum_i \T(\P_i)\le KT
  \right\}\] } 
Observe that each of the policies in the optimum
solution of \lagfirst\ will always choose an arm for the reward if the
observed value $q$ is larger than $\lambda(1)$. Therefore we can use
the same solution idea as Theorem~\ref{blahtheorem} and get a
collection of policies $\{\P'_i\}$ such that $\sum_i \T(\P'_i) \leq
KT$ and $\sum_{q>\lambda(1)} (q - \lambda(1)) N(q,\P'_i) \geq (\lagfirst - \lambda(1)
T)/(1+\epsilon')$ in time $O((\sum_i \edge(\S_i))\log (nT/\epsilon'))$.
Observe that this will involve using a second multiplier $\lambda(2)$
and choosing between $\lambda^+(2)$ and $\lambda^-(2)$ and the policies
$\P'_i$ will be randomized. 

\medskip We now perform (i) truncate $\P'_i$ to a horizon of $\beta T$
(denote these truncated policies as $\P''_i$) and then (ii) for each
$i$ let $\P_i$ be the policy $\P''_i$ with probability $\alpha \leq 1$
and otherwise be the null policy. Observe that we now have
$\{\h{\P}_i\}$ such that (note we are using Lemma~\ref{keymaxmab} on
$\tilde{R}(\P'_i)$ here):
\begin{enumerate}[(a)]
\item Each $\h{\P}_i$ has a horizon of at most $\beta T$; \label{cond1}
\item $\sum_i \T(\h{\P}_i) \leq \alpha KT$ and \label{cond2}
\item $\sum_i \sum_{q>\lambda_1} (q-\lambda(1)) N(q,\h{\P_i}) \geq \alpha\beta (\lagfirst - \lambda(1) T)/(1+\epsilon')$.\label{cond3}
\end{enumerate}
Now we observe that any feasible solution that
satisfies (\ref{cond1})--(\ref{cond3}) for $\lambda(1)$ continues to be
a feasible solution for $\lambda'(1) < \lambda(1)$.  We can now apply
the reasoning of Corollary~\ref{approxblahtheorem} and have a
collection of policies $\L_i$ that would satisfy the conditions
(\ref{cond1}),(\ref{cond2}), and the following:
{\small
\[
 \sum_i \sum_{q} N(q,\L_i) = \alpha\beta T/(1+\epsilon') < \alpha\beta T \qquad \mbox{and} \qquad 
\sum_i \sum_q q  N(q,\L_i) \geq \alpha\beta\lagfirst/(1+\epsilon')^2
\]
}
The lemma follows by setting $\epsilon'=\epsilon/3$
and observing that we will be computing the policies 
$\P_i$ at most $O(\log (nT/\epsilon'))$ due to the binary search. 
\end{proof}

\begin{definition}
Let $\lambda_i$ be the $\lambda(1) \in \{\lambda^+,\lambda^-\}$ chosen for $\L_i$.
\end{definition}

\subsection{The Final Scheduling of $\{\L_i\}$}
The scheduling policy for the \simult\ model is given in Figure~\ref{fhmaxmabtwo}.
\begin{figure*}[ht]
\fbox{
\begin{minipage}{6.0in}
{\small
\paragraph{Scheduling Policy \stall.} 
\begin{enumerate}
\item Each arm is in one three possible states: Ready, Current, and Finished. Initially, all arms are Ready. 
\item Make the first $K$ arms Current, and denote the set of
Current arms as $S$. 
\item We execute the policies $\Q_i$ for $i \in S$ as
described below. Whenever a policy terminates, we mark this arm as
Finished, remove it from $S$, and place any Ready arm in $S$.  At any
step, suppose the arms in $S$ are $i_1,\ldots,i_K$ in states
$u_1,\ldots,u_K$. Then the policy makes one of three decisions:

\begin{description}
\item[Full Step:] If $\sum_{s=1}^K \Pr[X_{i_{s}u_{s}} > \lambda_{i_s} ] \leq \frac{2}{3}$ then the policy plays all the arms. 
\item[Stalling Step:] If there exists any $s$ such that $\Pr[X_{i_{s}u_{s}} > \lambda_{i_s} ] \geq \frac13$ then we only play the arm $i_s$. 
\item[Throttling:] Find a subset $S'$ such that $\sum_{s \in S'} \Pr[X_{i_{s}u_s} \geq \lambda_{i_s} ] \in \left[\frac13, \frac23 \right]$ and schedule these arms.
\end{description}

\item At any step, the policy chooses the played arm whose observed reward
is maximum. For the purpose of analysis, we assume that any arm whose
observed reward is at least $\lambda$ can be chosen.
\end{enumerate} 
}
\end{minipage}
}
\caption{The Final Policy for \maxmab\ in the simultaneous feedback model\label{fhmaxmabtwo}}
\end{figure*}

\begin{lemma}
\label{lemtwo2}
The policy $\L_i$ executes to completion with probability at least 
$\left(1-\frac{3(1+\beta)\alpha}{1-\beta}\right)$ on {\em all} its decision paths.
\end{lemma}
\begin{proof}
  To bound the probability of this event, assume arm $i$ is placed
  last in the ordering, and only marked Current {\em after} all other
  policies are marked Finished - this only reduces the probability
  that $\L_i$ executes to completion on any of its decision paths. For
  the arms $j \neq i$, observe that at each step, either at least $K$
  arms are played, or sufficiently many arms are played so that the
  sum of the probabilities that the observed values exceed the respective 
  $\lambda_j$ values
  is at least $\frac13$. Therefore, at each step, if $S$ is the set of
  arms marked Current, $Z_{jt}$ is an indicator variable denoting
  whether arm $i$ is played, and $W_{jqt}$ is an indicator variable
  denoting whether arm $i$ is observed in state $q$, we must have:

{\small
\[ \sum_{j \in S} \left(\sum_{q > \lambda_j} W_{jqt} + \frac{Z_{jt}}{K} \right) \ge \frac13 \]
}
\noindent If $Y$ denotes the random time for which policies for arms $j \neq i$
execute, by linearity of expectation:
{\small
\[ \frac{\E[Y]}{3} \le \E\left[ \sum_t \sum_j \left(\sum_{q > \lambda_j} W_{jqt} + \frac{Z_{jt}}{K} \right) \right] =  \E\left[\sum_{j \neq i} \left( \sum_{q > \lambda_j}  N(q,\Q_j) + \frac{\T(\Q_j)}{K} \right) \right] \le \alpha\beta T + \alpha T \]
}
Therefore, $\E[Y] \le 3(1+\beta)\alpha T$ so that $\Pr[Y \le (1-\beta)T] \ge 
\left(1-\frac{3(1+\beta)\alpha}{1-\beta}\right)$ by
Markov's inequality. In this event, $\L_i$ executes to completion,
since its horizon is at most $\beta T$ due to truncation.
\end{proof}

\begin{lemma}
\label{lemthree3}
Consider the event $j \in S_0$ and $j$ is marked Finished. In this event, suppose we count the contribution from this arm to the overall objective only when it is the maximum played arm. Then,  its expected contribution is at least $\frac13$ times the value of policy $\L_j$.
\end{lemma}
\begin{proof}
  We have $\sum_i \Pr[X_{iu} \geq \lambda_i ] \leq \frac23$ whenever the
  arm $j$ is played simultaneously with any other arm. Since the other
  arms are independent of arm $j$, if $j$ is observed at $q \geq
  \lambda$ then with probability at least $\frac13$ all other arms are
  observed to be less than $\lambda$. The lemma follows.
\end{proof}

\noindent Combining Lemmas~\ref{truncate-within},~\ref{lemtwo2},
and~\ref{lemthree3}, we observe that with constant probability $\L_j$
executes to completion on all its decision paths.
Using linearity of expectation the combined simultaneous play policy 
has expected value at least $\frac{\alpha \beta}{3(1+\epsilon)}
\left(1-\frac{3(1+\beta)\alpha}{1-\beta}\right)$ times
\lpmaxmab. Setting $\alpha=(1-\beta)/(6(1+\beta))$ and
$\beta=\sqrt{2}-1$ the approximation ratio is less than $210$ for
suitable $\epsilon$. However even though the approximation factor is
large, the important aspect is that we are using the stronger bound
for an LP which captures the one-at-a-time model. This is summarized
in the next theorem.

\begin{theorem}
\label{harder:thm}
In time $O((\sum_i \edge(\S_i))\log^2 (nT/\epsilon))$ we can find a policy 
that executes in the \simult\ model and provides an
$O(1)$-approximation to the optimum policy in the \oneatatime\ feedback model.
\end{theorem}

\subsection{\maxmab\ in the \onlythemax\ Model}
\label{hardest}
Now consider the case where each arm has a budget $B_i$ and the budget
of an arm $i$ is depleted only if arm $i$ has the largest observed
value (and the decrease in the budget is that observed value). We
discuss the one-at-a-time feedback model for this problem. The
extension to the simultaneous play follows from a discussion
analogous to Section~\ref{harder} with appropriate loss of constants.

In this \onlythemax\ scenario, the state $u \in \S_i$ does not have sufficient
information to encode which plays so far resulted in the arm being
chosen (and hence the budget being depleted). Therefore, the state in
the single-arm dynamic program also needs to separately encode the
budget spent so far. This leads to difficulty in extending the results
for the non-budgeted case. As a concrete hindrance, we note that the
characterizations in Lemmas~\ref{lem:struct} and \ref{truncate-within} are 
no longer true for the Lagrangian, since the policy might decide not to choose an arm
with $q > \lambda$ (and deplete its budget) if it hopes that larger
values of $q$ will be subsequently observed. As a consequence of
these, the analysis for the budgeted case is significantly more
complicated.  Nevertheless, we show that an {\em approximately}
optimal solution to the Lagrangian has structure similar to 
Lemma~\ref{lem:struct} and hence, the scheduling policy in Figure~\ref{fhmaxmabone}
yields a constant approximation.

\begin{definition}
The policy $\P_i$ is considered {\bf p-feasible} if the total reward from choices made on any decision path  is at most $B_i$.
\end{definition}

As in the previous section, define the following linear program, where
the goal is to find a collection of {\em p-feasible} randomized
single-arm policies $\P_i$. As before, the policy can play the arm;
observe the reward, and decide to choose the arm; or stop execution.
The optimum value is bounded above by the LP relaxation below.

{\small
\[\achievedmaxmab = \mbox{Max}_{\mbox{p-feasible } \{\P_i\}} \left\{\sum_i  \sum_q q N(q,\P_i)\  | \ \sum_i  \sum_q   N(q,\P_i) \le T \mbox{ and } \sum_i \T(\P_i) \le KT  \right\} \]
}
Observe that the system \achievedmaxmab\ differs from \lpmaxmab\ in just the definition of the feasible policies. To simplify the remainder of the 
presentation, we begin with the assumption (but it is removed later).

\begin{enumerate}[({A}1)]
\item All reward values $q$ are powers of two. This assumption loses a  factor of $2$ in the approximation ratio, since we can round $q$  down to powers of $2$ while symbolically maintaining their  distinction in the prior updates. \label{firsta}
\end{enumerate}

\noindent We show the following surprising Lemma:

\begin{lemma}
\label{coollemma}
Suppose we are given a {\bf p-feasible} policy $\P^p_i$ and a cost of choosing a reward of $\lambda(1)$; then there exists a policy $\P_i$ in the \allpay\ model which satisfies $\T(\P_i) \leq \T(\P^p_i)$ and 
{\small
\[
 \sum_q   (q-\lambda(1)) N(q,\P_i) \geq \frac14 \sum_q   (q-\lambda(1)) N(q,\P^p_i)
\]
} 
under assumption (A1).
\end{lemma}
\begin{proof}
  Define $q$ to be small if $q \in [\lambda(1),2\lambda(1))$ and large
  otherwise. Since $q$ values are powers of $2$, there is only one
  small value of $q$; call this value $q_s$.  Consider the policy $\P^p_i$.
  Either half the objective value $\sum_q   (q-\lambda(1)) N(q,\P^p_i)$ is 
  achieved by choosing $q_s$ ({\bf Case 1}) or
  by choosing large $q > q_s$ ({\bf Case 2}). 
  The policy $\P_i$ will be different depending on the two cases ---
  the algorithm need not know which case applies; we will solve both
  cases and choose the better solution. In {\bf Case 1}, we set $\nu = \lambda(1)$, 
  and {\bf Case 2} set $\nu = \min\{q | q \ge 2\lambda(1)\}$.  

\medskip
  Modify policy $\P^p_i$ as follows: In the new policy $\P_i$,
  whenever a reward observed is at least $\nu$ choose the arm even if
  the original policy had not chosen it. Note that we immediately have
  $\T(\P_i) \leq \T(\P^p_i)$ since we are not affecting the play
  decisions.

  \medskip We now bound the value of the new policy by considering
  every decision path of the original policy by fixing a decision path
  and consider the contribution to the objective $\sum_q
  (q-\lambda(1)) N(q,\P_i)$ in this decision path.

  \smallskip\noindent{\bf Case 1:} Suppose choosing $q_s$ contributes more than
  half of the value $\sum_q (q-\lambda(1)) N(q,\P^p_i)$. In this case,
  $\nu = \lambda(1)$. Consider just the contribution of $q_s -
  \lambda(1)$ on any decision path, and suppose $q_s$ is chosen $k$
  times, so that $k q_s \le B_i$. If the modified policy $\P_i$ chooses this
  $q$ for $k'$ times where $k' \ge k$ times, then its value clearly 
  dominates the contribution to $\P^p_i$. Otherwise, suppose $\P_i$ 
  chooses $q_s$ some $k' < k$ times, then the choices of the remaining
  $q$ must have exhausted the budget. Since per unit $q$, the policy
  generates value at least $1 - \frac{\lambda(1)}{q_s}$, a knapsack
  argument shows that the value generated in the modified policy on
  all $q$ is at least the value generated by the original policy on
  the $q_s - \lambda(1)$ values:

{\small
\[ \mbox{Value of $\P_i$ on the path} \ge  B_i \left(1 - \frac{\lambda(1)}{q_s} \right)  \ge kq_s \left(1 - \frac{\lambda(1)}{q_s} \right) =   k(q_s - \lambda(1)) \]
}
\noindent Note that contribution from $q_s$ in $\P^p_i$ to the path is $k(q_s - \lambda(1))$.

\smallskip\noindent
{\bf Case 2:} Suppose the large $q$ values contribute to at least 
half of the value $\sum_q (q-\lambda(1)) N(q,\P^p_i)$.
In this case, we have $\nu = \min\{q
| q \ge 2\lambda(1)\}$. On any decision path, consider just the
contribution of large $q - \lambda(1)$.  Note that $q - \lambda(1) \ge q/2$
for such $q$, so that in the worst case, the new policy exhausts its
budget and contributes $B_i/2$ to the objective while the original
policy could have contributed $B_i$. If the new policy does not
exhaust its budget, it must match the value of the original policy on
large $q$.

\smallskip\noindent
Since we lose a factor two in splitting the analysis into two cases,
and a factor of $2$ within the second case. Therefore, the new policy
is a $4$ approximation and satisfies the properties in the statement.
\end{proof}

\noindent The next Lemma follows immediately:

\begin{lemma}
\label{truncate-within2}
For $\beta,\alpha \in (0,1]$ in time $O((\sum_i \edge(\S_i))\log^2 (nT/\epsilon'))$ 
we can find a collection of single arm policies 
$\{\L_i\}$ in the \allpay\ model such that (a) Each $\L_i$ has a horizon of at most $\beta T$;
(b) $\sum_i \T(\L_i) \leq \alpha KT$;
(c) $\sum_i \sum_{q} N(q,\L_i) \leq \alpha\beta T/4$; and
(d) $\sum_i \sum_q q  N(q,\L_i) \geq \alpha\beta\achievedmaxmab/(4(1+\epsilon))$, under the condition (A1). The condition (A1) can be removed if we relax condition (d) to $\sum_i \sum_q q  N(q,\L_i) \geq \alpha\beta\achievedmaxmab/(8(1+\epsilon))$.
\end{lemma}
\begin{proof}
The proof is identical to the proof of Lemma~\ref{truncate-within} except in the stage where we find the policies $\P'_i$, where we can only guarantee 
{\small
  \[ \sum_{q>\lambda(1)} (q - \lambda(1)) N(q,\P'_i) \geq
  \frac14(\lagfirst - \lambda(1) T)/(1+\epsilon') \] } 
as a
consequence of Lemma~\ref{coollemma}. Note that to apply
Lemma~\ref{coollemma} we need to use two different dynamic programming
solutions --- one for considering all values larger than $q_s$ and one
for all values larger or equal to $2q_s$. Observe that $q_s$ is fixed
when we know $\lambda(1)$.
\end{proof}

\noindent We are ready to prove the main result of this subsection.

\begin{theorem}
\label{hardest:thm}
In time $O((\sum_i \edge(\S_i))\log^2 (nT/\epsilon))$ we can find a policy 
that executes in the \allpay\ model and provides an
$O(1)$-approximation to the optimum policy in 
the \onlythemax\ model, where both are measured in the \oneatatime\ feedback model.

In the same running time we can find a policy that executes in the conjunction of 
\allpay\ and \simult\ models and guarantees an $O(1)$-approximation to the 
\onlythemax\ accounting in the \oneatatime\ feedback model. Therefore 
the optimum solutions in all the
possible four models are within $O(1)$ factor of each other.
\end{theorem}
\begin{proof}
  The first part of the theorem follows from finding policies
  $\{\L_i\}$ as described in Lemma~\ref{truncate-within2} setting
  $\alpha=1,\beta=1$ and observing that $\sum_i \left(\sum_q N(q,\L_i)
    + \T(\L_i)/K\right) \leq 5T/4$. We can now run the scheduling
  policy as described in Figure~\ref{fhmaxmabone} using the
  probability $4/5$ in Step~\ref{steptobereffed}. It is easy to
  observe that the reward of the combined policy is at least
  $\frac25\sum_i \sum_q q N(q,\L_i)$. This is at least
  $\achievedmaxmab/(10(1+\epsilon))$ under the assumption (A1) and
  $\achievedmaxmab/(20(1+\epsilon))$ otherwise. Observe that the policy we execute
  is in the \allpay\ model.

  For the second part we again use Lemma~\ref{truncate-within2} and
  run the scheduling algorithm in Figure~\ref{fhmaxmabtwo} and the
  expected reward is $\frac{\alpha \beta}{24(1+\epsilon)}
  \left(1-\frac{3(1+\beta/4)\alpha}{1-\beta}\right)$ which is $O(1)$
  for suitable $\alpha,\beta$.
  The last part of the theorem now follows.
\end{proof}

\section{Future Utilization and Budgeted Learning}
\label{sec:budget}
\renewcommand{\R}{{\tilde{R}}}
\renewcommand{\Q}{{\tilde{Q}}}
In this section, we consider the budgeted learning problem. We assume
$K=1$ arms can be played any step. The goal of a policy here is to
perform pure exploration for $T$ steps. At the $T+1^{st}$ step, the
policy switches to pure exploitation and therefore the objective is to
optimize the expected reward for the $T+1^{st}$ play.  More formally,
given any policy $P$, each decision path yields a final (joint) state
space. In each such final state, choose that arm $i$ whose final state
$u \in \S_i$ has maximum reward $r_u$; this is the reward the policy
obtains in this decision path. The goal is to find the policy that
maximizes the expected reward obtained, where the expectation is taken
over the execution of the policy. We provide a
$(3+\epsilon)$-approximation for the problem in time $O((\sum_i
\edge(\S_i))\log (nT/\epsilon))$ time. The result holds for adversarially 
ordered versions, where an adversary specifies the order in which we can 
execute the policies without returning to an arm. The result extends to a
$4+\epsilon$-approximation under metric traversal constraints.
\eat{
However one of the main reasons for presenting this problem now is to
introduce an amortization technique of rewards for policies where the
reward is not a sum of the rewards of the plays. This technique is
introduced in Section~\ref{sec:amort2}, and is useful for the next section 
as well where we consider the \maxmab\ problem. 
}
\paragraph{Single arm policies.} Consider the execution of a policy
$P$, and focus on some arm $i$. The execution for this arm defines a
single arm policy $P_i$, with the following actions available in any
state $u \in \S_i$: (i) Make a play; (ii) Choose the arm as final, and
obtain reward $r_u$; or (iii) Quit.

\begin{definition}
Given a single arm policy $P_i$ for arm $i$ let $\I(P_i)$ be the probability that arm $i$ is chosen as final. Let $\R(P_i)$ be the expected reward obtained from events when the arm is chosen as final -- note that this is different from previous sections. As before $\T(P_i)$ is the number of expected plays made by the policy $P_i$.
\end{definition}

\subsection{A Weakly Coupled Relaxation and a $(3+\epsilon)$-Approximate Solution}
\noindent We define the following weakly coupled formulation:
\[ \lpbud=\mbox{Max}_{\mbox{Feasible } \{P_i\}} \left \{ \sum_i \R(P_i) \left| \sum_i \left( \I(P_i) + \T(P_i)/T \right) \right. \leq 2 \right\} \]
\begin{lemma}
 \lpbud\ is an upper bound for the reward of the optimal policy, $OPT$. 
 \end{lemma}
 \begin{proof}
Consider the (randomized) single-arm policies obtained by viewing the execution of the optimal policy restricted to each arm $i$.  It is easy to check via linearity of expectation that these policies  are feasible for the constraint in \lpmaxmab, since the expected number of choices made on any decision path is at most one, and the expected number of plays made is at most $T$. Further, the objective value of these policies is precisely $OPT$.
 \end{proof}

In what follows we show that we can construct a feasible policy which has an objective value $\lpbud/3$. This proof will be similar to the argument in Section~\ref{sec:metric}. Consider the Lagrangian:
\[ \budlag=\mbox{Max}_{\mbox{Feasible } \{P_i\}} \left \{ 2\lambda + \sum_i \left( \R(P_i) - \lambda \I(P_i) - \lambda \T(P_i)/T \right) \right\} = 2\lambda + \sum_i 
\Q_i(\lambda)
\]
\noindent where $\Q_i(\lambda)$ is the maximum value of $\R(P_i) - \lambda \I(P_i) - \lambda \T(P_i)/T$ for a feasible policy $P_i$. Immediately it follows from weak duality that (compare Lemma~\ref{lem:opti} and Lemma~\ref{mainlem});

\begin{lemma}
\label{lem:opti2}
$\Q_i(\lambda)$ can be computed by a dynamic program in time $O(\edge(\S_i))$. For the optimum policy $P_i(\lambda)$ 
we have $\R(P) = \Q_i(\lambda) + \lambda (\I(P_i(\lambda)) + \T(P_i(\lambda))/T)$.
Moreover $\I(P_i(\lambda)) + \T(P_i(\lambda))/T$ is non-increasing in $\lambda$.
\end{lemma}
\begin{proof}
Let $Gain'(u)$ to be the maximum of the objective of the single-arm policy conditioned on
starting at $u \in \S_i$. We perform a bottom up dynamic programming over $\S_i$.
If $u$ has no children, then if we ``choose the node as a final answer'' at
node $u$, then $Gain'(u)= r_u -\lambda$ in this case. Stopping and not doing anything
corresponds to $Gain'(u)=0$. Playing the arm at node $u$ will rule out either of the outcomes, and in this case $Gain'(u) = -\frac{\lambda}{T} + \sum_v \p_{uv} Gain'(v)$. In summary:
\[ Gain'(u) = \max \left\{ 0, \ \ r_u -\lambda, \ \ -
  \frac{\lambda}{T} + \sum_v \p_{uv} Gain'(v) \right\} \] 
We have $Gain'(\rho_i)=\Q_i(\lambda)$. The running time follows from inspection.  Observe that $\R(P_i(\lambda)) = \Q_i(\lambda) + \lambda
(\I(P_i(\lambda)) + \T(P_i(\lambda))/T)$ is maintained for every
subtree corresponding to the optimum policy starting at $u \in \S_i$.
Note that increasing $\lambda$ decreases the contribution of the
latter two terms which corresponds to nonincreasing $\I(P),\T(P)$.
\end{proof}
\begin{lemma}
\label{lem:bal2}
$2\lambda + \sum_i \Q_i(\lambda) \geq \lpbud$. Moreover for any
constant $\delta > 0$,  we can choose a $\lambda^*$ in polynomial time so that for the resulting collection of policies $\{P^*_i\}$ which achieve the respective $\Q_i(\lambda^*)$ we have: (1) $\lambda^* T \ge (1-\delta)\lpbud/3$ and (2) $ \sum_{i} \Q_i(\lambda^*) \ge \lambda^*$.
\end{lemma}
\begin{proof}
The first part follows from weak duality. For the second part, our approach is similar to that in Theorem~\ref{blah2}. Observe that at very large $\lambda$ we have $\Q_i(\lambda)=0$ for all $i$ and therefore $\sum_i\Q_i(\lambda) < \lambda$. At $\lambda=0$ we easily have $\sum_i \Q_i(\lambda) \geq \lambda$. we can now perform a binary search such that 
we have two values of $\lambda^{-},\lambda^{+}$  maintaining the invariants that $\sum_i \Q_i(\lambda^{-}) \geq \lambda^{-}$ and 
$\sum_i \Q_i(\lambda^{+}) < \lambda^{+}$ respectively. Note that it suffices to ensure $\lambda^+ - \lambda^- \leq \delta \lambda^-/3$; where we can use $\lambda^*=\lambda^-$. The lemma follows from that fact that
$2\lambda + \sum_i \Q_i(\lambda) \geq \lpbud$.
\end{proof}

\subsection{An Amortized Accounting}
\label{sec:amort2}

Consider the single-arm policy $\P_i(\lambda^*)$ that
corresponds to the value $\Q_i(\lambda)$.  
$\P_i(\lambda^*)$ performs one of
three actions for each state $u \in \S_i$: (i) Play the arm; (ii) choose the arm as the final answer and stop or (iii) Quit. 
For this policy, 
$\R(P_i(\lambda^*)) = \Q_i(\lambda^*) + \lambda^*
(\I(P_i(\lambda^*)) + \T(P_i(\lambda^*))/T)$.
This implies that the reward
$R(\P_i(\lambda^*))$ of this policy can be amortized, so that for
state $u \in \S_i$, the reward is collected as follows:

\begin{enumerate}\parskip=0in 
\item An upfront reward of $\Q_i(\lambda^*)$ when the play for the arm initiates at the root $\rho \in \S_i$.
\item A reward of $\lambda^*$ for choosing the arm in any state $u$ and stopping.
\item A reward of $\lambda^*/T$ for playing the arm in $u \in \S_i$.
\end{enumerate}
{\em Note that this accounting is not true if the policy $\P^*_i(\lambda)$ is executed incompletely}, for instance, if it is terminated prematurely. 

\subsection{The Final Algorithm}

\begin{figure*}[htbp]
\fbox{
\begin{minipage}{6.0in}
{\small
{\bf Policy {\sc Final}$(\lambda^*)$ for Budgeted Learning}
\begin{enumerate}
\item Define the problem \lpbud\ and solve the Lagrangian \budlag\ for $\lambda^*$ as chosen in Lemma~\ref{lem:bal2}. This yields single arm policies $\{\P^*_i(\lambda)\}$.
\item Order the arms arbitrarily (or as provided by an adversary) as $1,2,\ldots,n$ and execute the policies $\P_i(\lambda^*)$ in this order, with the stopping conditions: 
\begin{enumerate}
\item If $\P^*_i(\lambda)$ stops, move to the next arm.
\item If $\P^*_i(\lambda)$ chooses arm $i$ as final, choose arm $i$ as final (obtaining its reward) and stop.
\item If $T$ steps have elapsed, choose the current arm as final (obtaining its reward) and stop.
\end{enumerate}
\end{enumerate}
}
\end{minipage}
}
\caption{The Policy {\sc Final}$(\lambda^*)$.\label{ten}}
\end{figure*}

We now present the overall algorithm in Figure~\ref{ten}.
The next lemma is the crux of the entire analysis, and follows the amortized accounting argument. The hurdle with using the argument directly is that when the horizon is exhausted, the currently executing single-arm policy is not executed completely. We first pretend it executed completely, violating the horizon.

\begin{lemma}
\label{lem:violate}
Suppose {\sc Final}$(\lambda^*)$ completely executes the single-arm policy that it is executing at time $T$, before stopping. (Observe that this policy is not feasible.) The expected reward of this (infeasible) policy is at least  $\lpbud/(3+\epsilon)$.
\end{lemma}
\begin{proof}
Consider three stopping conditions: (a) The policy has visited all the arms and have no further arms to play,  (b) the policy choose an arm, or (c) the policy continue past $T$ steps.

By the amortized argument, in case (a) the contribution to the reward using the amortized accounting is $\sum_i \Q(\lambda^*)$ which is at least $\lambda^*$ by Lemma~\ref{lem:bal2}. In case (b), the contribution is again $\lambda^*$. In case (c) the contribution is $T \cdot \lambda^*/T = \lambda^*$. Thus in all cases the contribution is at least $\lambda^*$ and the Lemma follows (setting $\delta \leq \epsilon/3$ in Lemma~\ref{lem:bal2}).
\end{proof}

\begin{theorem}
\label{thm:order}
The policy {\sc Final}$(\lambda^*)$ achieves a reward of at least $\lpbud/(3+\epsilon)$.
\end{theorem}
\begin{proof}
 The policy {\sc Final}$(\lambda^*)$ differs from that in Lemma~\ref{lem:violate} as follows: At the $T+1^{st}$  step, instead of continuing executing the current policy $\P_j(\lambda^*)$ for some arm  $j$ in state $u \in \S_j$ (and continuing the amortized accounting), {\sc Final}$(\lambda^*)$ simply chooses the arm $j$ as the final answer. Let the remainder  of the policy which was not executed be $P$.

By the martingale property of the rewards, choosing the arm $j$ at this state  $u$ would contribute to the objective at least $\R(P)$ because $P$ may not choose the arm at all and have $0$ contribution in some  evolutions. Therefore, stopping single-arm policy $\P_j(\lambda^*)$ prematurely (when the horizon is exhausted) and choosing it as final yields at least as much reward as executing it completely. Therefore, the reward of {\sc Final}$(\lambda^*)$ is at least the reward of the infeasible policy analyzed in Lemma~\ref{lem:violate}, and the theorem follows.
\end{proof}

\noindent The next corollary follows by inspection based on the arguments leading up to
  Theorem~\ref{thm:order}.

\begin{corollary} If we can solve $\sum_i \Q_i(\lambda)$ to an   $\alpha$ approximation then we would have an $(\alpha+2+\epsilon)$  approximation for budgeted learning.
\end{corollary}

Observe that the Lagrangian formulation in \budlag\ would still
satisfy Lemma~\ref{prop:cool}, and as a consequence we immediately
have a factor $(\alpha+2+\epsilon)$-approximation algorithm for the
variant of the budgeted learning problem where we have a metric
switching cost and a total bound of $L$ on the total switching cost,
following the $\alpha=2+\epsilon$ (or possibly other) approximation
algorithms for the {\em Orienteering} problem discussed in
Section~\ref{sec:metric}.
We can also solve the budgeted learning problem with $r$
additional packing (knapsack type) constraints to within $(r + 2 +
\epsilon)$-approximation.

\section{Conclusions}

In this paper, we have shown that for several variants of the finite
horizon multi-armed bandit problem, we can formulate and solve {\em
  weakly coupled} LP relaxations, and use the solutions of these
relaxations to devise feasible decision policies whose reward is
within a fixed constant factor of the optimal reward. This provides
analytic justification for using such relaxations to guide policy
design in practice, and the resulting policies are comparable in
complexity to standard index policies.

The main open questions posed by our work is to improve the performance
bounds for large delays in the delayed feedback model or the results for 
the \simult\ model for the \maxmab\ problem. Observe that in the latter case
we are providing comparisons against the strong upper bounds of
\oneatatime\ model. This would require new techniques or 
a lower upper bound.
This could involve formulation of more
strongly coupled LPs, for instance, polymatroidal
formulations~\cite{BN96,nino,nino2} or time-indexed
formulations~\cite{Gupta}. It would also be interesting to
characterize the class of bandit problems for which weakly coupled
relaxations provide constant factor approximations.

\section*{Acknowledgments} 
We thank Martin P\'al for several helpful discussions.

\bibliographystyle{abbrv}

\end{document}